\documentclass[12pt]{article}
\usepackage{amsmath}
\usepackage{amsfonts} 
\usepackage{graphicx,psfrag,epsf}
\usepackage{enumerate}
\usepackage{natbib}
\usepackage{url} 
\usepackage{hyperref}
\usepackage{graphicx}
\usepackage{amsthm}
\usepackage{bbm}
\usepackage{soul}
\usepackage{paralist}
\usepackage{mathtools}
\usepackage{adjustbox}
\usepackage{booktabs}
\usepackage{xcolor}
\usepackage{amssymb}
\usepackage{multirow}  
\usepackage{array}     
\usepackage{physics}
\usepackage{tablefootnote}
\usepackage{enumitem}
\usepackage{amssymb}

\usepackage{lipsum}

\allowdisplaybreaks

\newtheorem{theorem}{Theorem}

\newtheorem{assumption}{Assumption}

\theoremstyle{definition}

\theoremstyle{remark}
\newtheorem{remark}{Remark}

\newcommand{\PP}{\mathbb{P}}

\newcommand{\indep}{\perp\!\!\!\perp}
\newcommand{\calS}{\mathcal{S}}
\newcommand{\calX}{\mathcal{X}}

\newcommand{\E}{\mathbb{E}}

\newcommand{\Z}{\mathbf{Z}}

\newcommand{\corr}{\mathbb{C}\mathrm{or}}

\newcommand{\Var}{\mathbb{V}{\rm ar}}

\usepackage{setspace}
\usepackage{bibspacing}
\newcommand{\blind}{0}

\begin{document}

\def\spacingset#1{\renewcommand{\baselinestretch}%
{#1}\small\normalsize} \spacingset{1}


\if0\blind
{
  \title{\bf Causal sensitivity analyses for indirect and spillover mediation effects in cluster-randomized trials: Application to a pragmatic trial on chronic pain management}
  \author{Yuki Ohnishi$^{1}$, Michael J. Daniels$^{1}$ and Fan Li$^{2}$\\
    $^1$Department of Statistics, University of Florida, \\
    $^2$Department of Biostatistics,  
    Yale University School of Public Health
    }
\date{}  
  \maketitle
} \fi

\if1\blind
{
  \bigskip
  \bigskip
  \bigskip
  \begin{center}
    {\LARGE\bf }
\end{center}
  \medskip
} \fi

\begin{abstract}
Causal mediation in cluster-randomized trials can clarify whether a cluster-level intervention changes patient outcomes through patient- and peer-level processes, but identification relies on untestable cross-world assumptions. We develop a two-part sensitivity analysis for violations of two distinct cross-world independence assumptions. The first component adapts bridge score sensitivity analysis to assess residual mediator-outcome confounding with clustered mediator exposures; the second uses a Gaussian copula to assess cross-world dependence between an individual's treated mediator and summaries of control-arm peer mediators. The resulting framework parameterizes unidentifiable dependence structures using interpretable correlation parameters and combines them with a flexible Bayesian nonparametric model to characterize the plausible range of the target estimands. A nested common-atoms enriched Dirichlet process models the observed data and supports posterior-predictive standardization. We apply the framework to the Pain Program for Active Coping and Training trial, a primary care-embedded cognitive behavioral therapy intervention for adults with chronic pain receiving long-term opioid therapy. The intervention had a negative total effect on twelve-month pain-related disability. The natural indirect and direct effect decomposition was sensitive to plausible directional residual confounding, whereas the spillover and individual mediation effect decomposition changed little across the copula sensitivity analyses and attributed most of the negative indirect signal to the individual's mediator. Jointly relaxing both assumptions substantially increased uncertainty about the pathways. The analysis supports an overall improvement in disability and a plausible individual pain-impact pathway, while showing that detailed mechanistic claims remain assumption-dependent.
\end{abstract}

\noindent%

\spacingset{1.9} 

\section{Introduction}

Chronic pain is one of the most common, disabling, and difficult conditions managed in primary care. Many patients have pain at multiple sites together with medical and mental health comorbidities, making long-term management especially challenging. Opioids were widely used for chronic pain despite limited evidence for sustained benefit and substantial patient and public-health harms. As a result, primary care systems need practical nonopioid options for patients already receiving long-term opioid therapy. The Pain Program for Active Coping and Training (PPACT) trial \citep{debar2022ppact} was a pragmatic cluster-randomized trial (CRT) of a primary care-embedded, interdisciplinary cognitive behavioral therapy (CBT) intervention for adults with chronic pain receiving long-term opioid therapy \citep{debar2022ppact}. The trial randomized 106 primary care provider (PCP) clusters to CBT or usual care before patient recruitment. The intervention included an intake evaluation and twelve weekly group sessions delivered by an interdisciplinary team in collaboration with the PCP.

The main PPACT report found that the CBT intervention produced modest but sustained reductions in pain impact and pain-related disability compared with usual care, while showing little evidence of reduced opioid dose. Having established effectiveness, the natural applied question is how the intervention worked. The most direct pathway is individual, e.g., CBT intervention may reduce later disability by helping patients manage pain impact earlier. Additionally, PPACT's provider (i.e. cluster) randomized design makes the pathway potentially broader. Patients shared clinicians and health-system contexts and group sessions. Therefore, changes in peer context, clinic practice, or the care environment could also contribute to outcomes, leading to potential \emph{spillover} effects on patients’ disability channeled through peers' improvement on intermediate outcomes. 

Causal mediation analysis provides a useful framework for this structural question. 
The causal mediation literature formalized natural direct and indirect effects, their identifying assumptions, and sensitivity analyses for mediator--outcome confounding in independent-data settings \citep{Imai2010,Tchetgen2012}. Extensions to multilevel settings and group-level exposures defined direct and indirect effects for clustered populations \citep{VanderWeele2009}. For CRTs, \citet{VanderWeele2013_CRT} showed that mediated effects can contain both within-cluster and spillover components when one unit's mediator affects another unit's outcome. More recent work has developed Bayesian nonparametric and semiparametric estimators for natural, individual, and spillover mediation effects in CRTs \citep{ohnishi2025,Cheng2026,ohnishi2026identification}. 
In the PPACT CRT, the six-month PEGS pain impact score is a clinically meaningful intermediate response, and the twelve-month Roland--Morris Disability Questionnaire (RMDQ) score captures later pain-related disability. Mediation analysis in this trial would then ask how much of the intervention effect on later disability is consistent with an earlier improvement in pain impact, how much remains outside that measured pathway, and whether the mediated part appears to operate mainly through the patient's own pain impact or through the pain impact experience of peers in the same provider patient panel. 
Untangling these individual and spillover pathways can guide future implementation by indicating what should be strengthened, monitored, or redesigned when primary-care systems scale similar CBT programs for improving chronic pain management.

However, these mediation analyses are more assumption-dependent than the main effectiveness analysis. Cluster randomization supports the comparison of CBT with usual care, but it does not randomize the six-month pain impact. To interpret pain impact as a mediator, the analysis must address the possibility that unmeasured factors jointly shape pain impact and later disability even after adjustment for baseline information. In PPACT, this concern is clinically salient because both the intermediate measure and the later outcome are patient-reported pain measures. Various unmeasured factors, such as pain catastrophizing, session engagement, concurrent nonpharmacologic therapies, social support, or short-term life stressors, could affect both the six-month pain impact score and the twelve-month disability score.

A second assumption is CRT-specific and arises when the mediated pathway is split into individual and spillover components. That split asks what would happen if a patient's own pain impact response followed the CBT arm while the peer context followed usual care. Such combinations cannot be observed directly in the trial. Existing mediation approaches for CRTs commonly treat the relevant patient and peer responses as independent after conditioning on observed information \citep{ohnishi2025,Cheng2026}. In PPACT, however, provider clusters may share unmeasured practice styles, pain-management norms, regional treatment cultures, or group dynamics that connect patients' responses within the same cluster. These shared features matter for an applied interpretation because they determine whether the apparent mechanism is mainly a patient's own improvement or a broader peer and clinic-level process.

Prior work has developed several sensitivity analysis frameworks for causal mediation analysis, but is strictly restricted to the independent data setting. \citet{VanderWeele2010_bias} developed bias formulas and effect decompositions, complemented by sensitivity analyses \citep{VanderWeele2014_effect_decomposition} and approaches based on bivariate normal latent confounders \citep{le_Cessie_2016,Lindmark2018}. \citet{Ding2016_sharp}, along with follow-up work \citep[e.g.,][]{SjolanderWaernbaum_2024}, established sharp bounds for risk ratios under interpretable parameters characterizing the strength of confounding. Bayesian approaches include \citet{mccandless2017,Daniels_Linero_Roy_2023}, which emphasize prior factorization into likelihood and sensitivity components. More recently, \citet{ohnishi2026sensitivityanalysiscausalmediation} introduced the bridge score, together with sharp bounds and a calibration framework, to assess violations of mediator ignorability.
Despite these advances, principled sensitivity analysis for the clustered settings allowing for spillover mediation effects remains underexplored. In particular, they do not address dependence structures unique to clustered data, such as violations of cross-world inter-individual mediator independence.

This paper studies the intervention mechanism of the PPACT trial by developing a sensitivity analysis framework for these unverifiable assumptions. Methodologically, the analysis separates two sources of uncertainty that are often blended together in mechanistic trial reports. First, we adapt bridge score sensitivity analysis \citep{ohnishi2026sensitivityanalysiscausalmediation} to a clustered mediator exposure that includes both a patient's own pain impact and the pain impact experience of peers in the same provider cluster. This layer calibrates residual mediator-outcome confounding on the outcome scale, allowing the indirect and direct pathways to be reevaluated under transparent assumptions about unmeasured confounders. Second, we adopt a Gaussian copula for the sensitivity analysis of within-cluster dependence, focusing on the key unidentified feature, i.e., the dependence between a patient’s own potential mediator response and the peer responses that would arise under the alternative treatment condition. Together, these two layers provide a modular workflow for CRT mediation analysis, with transparent sensitivity checks for potential violations of unverifiable assumptions. The observed trial data are used to learn the intervention, mediator, and outcome relationships, whereas the unidentified components of the pathway decomposition are varied through interpretable sensitivity parameters.

The empirical results support a cautious clinical message. Under the default analysis with no sensitivity correction, the PPACT intervention has a credibly beneficial effect on twelve-month disability, and the decomposition assigns beneficial contributions to both the pain impact pathway and the remaining direct pathway. Sensitivity analysis shows, however, that this indirect/direct split can change under plausible directional residual confounding between pain impact and disability. The spillover analysis is more stable in this pragmatic trial. Varying the unidentified dependence between a patient's treated pain impact and either the control peer mean or peer variability changes the spillover estimate only modestly, and the mediated signal is carried primarily by patients' own pain impact rather than by the peer-distribution pathway. Thus, PPACT provides evidence of improved disability and a plausible individual pain impact pathway, but the more detailed mediation claims should be interpreted through the sensitivity analyses rather than as assumption-free consequences of the trial.

\subsection{Pain Program for Active Coping and Training (PPACT) trial}
\label{sec:motivating_example}

The analytic PPACT dataset contains 850 patients from 106 PCP clusters across three Kaiser Permanente regions \citep{debar2022ppact}. We use the six-month PEGS pain impact score as the mediator and the twelve-month RMDQ score as the outcome, matching the clinical ordering in which CBT may first change pain impact and later change pain-related disability.
Table~\ref{tab:ppact_table1} summarizes the selected analytic variables. The mediator differs more strongly by arm than the outcome: mean 6-month PEGS is 5.50 under CBT versus 6.20 under usual care, whereas mean 12-month RMDQ is 0.61 versus 0.68. This pattern makes mediation analysis substantively useful, but also fragile. That is, both PEGS and RMDQ are pain measures reported by patients and may be jointly affected by partially observed factors such as pain coping, catastrophizing, engagement, social support, or concurrent non-pharmacologic therapies. Follow-up responses used in the analysis were completed during preprocessing, but missing baseline covariates were retained in the modeling input and updated from their posterior predictive distribution under the proposed Bayesian nonparametric model rather than filled with a fixed mean or median.

\begin{table}[!htbp]
\centering
\small
\caption{Selected summary of the analytic PPACT dataset by treatment arm. Entries are counts (percentages) unless otherwise noted; descriptive entries use the completed Table~1 dataset, while the CA-EDP analysis retains missing baseline covariates for posterior updating. The mediator is the 6-month PEGS pain impact score and the outcome is the 12-month Roland--Morris Disability Questionnaire score.}
\label{tab:ppact_table1}
\begin{adjustbox}{width=\textwidth}
\begin{tabular}{lrr}
\toprule
Characteristic & Usual care (417 patients) & CBT intervention (433 patients) \\
\midrule
\multicolumn{3}{l}{\textbf{Study design}} \\
No. of participants & 417 & 433 \\
No. of PCP clusters & 53 & 53 \\
Cluster size: median [min, max] & 8 [3, 13] & 9 [3, 11] \\
\addlinespace
\multicolumn{3}{l}{\textbf{Selected baseline characteristics}} \\
Age (years), mean (SD) & 59.12 (12.23) & 61.35 (11.86) \\
Women & 286 (68.6\%) & 287 (66.3\%) \\
BMI (kg/m\textsuperscript{2}), mean (SD) & 32.59 (8.64) & 32.83 (9.12) \\
2 or more chronic comorbid medical conditions & 156 (37.4\%) & 148 (34.2\%) \\
Anxiety and/or depression & 189 (45.3\%) & 181 (41.8\%) \\
No. of chronic pain conditions, median [IQR] & 4 [3, 5] & 4 [3, 5] \\
Average daily opioid dose (MME), median [IQR] & 30.83 [15.98, 75.06] & 28.50 [16.00, 53.17] \\
Pain impact PEGS score (0--10), mean (SD) & 6.65 (1.77) & 6.50 (1.69) \\
Pain-related disability RMDQ (0--1), mean (SD) & 0.69 (0.19) & 0.67 (0.19) \\
\addlinespace
\multicolumn{3}{l}{\textbf{Mediator and outcome}} \\
PEGS pain score at 6 months, mean (SD) & 6.20 (2.04) & 5.50 (2.11) \\
RMDQ disability score at 12 months, mean (SD) & 0.68 (0.20) & 0.61 (0.24) \\
\bottomrule
\end{tabular}
\end{adjustbox}
\end{table}

\section{Notation and setup}
\label{sec:bridge_score_sa}

\subsection{Observed data and potential outcomes}

Consider a clustered-randomized trial (CRT) with clusters indexed by $i=1,\dots,I$, and let unit $j \in \{1,\dots,N_i\}$ belong to cluster $i$ with a cluster size $N_i$. Let $A_i\in\{0,1\}$ denote the cluster-level treatment assignment. We collect the cluster-level baseline covariates in $V_i$ and the individual-level baseline covariates in $C_{ij}$. Define
$X_{ij}=(V_i,C_{ij})\in\cal X$ as the combined pretreatment covariate vector relevant to unit $j$, $\underline C_i=(C_{i1},\ldots,C_{iN_i})$ as the vector of individual-level covariates in cluster $i$, and $X_i=(V_i,\underline C_i)$ as the full cluster covariate collection. Thus $X_{ij}$ indexes unit-level conditional laws, whereas $X_i$ is used only when the joint law of all units in a cluster is required. When cluster and unit subscripts are suppressed, $X$ denotes a generic draw of $X_{ij}$. Let $\underline M_i=(M_{i1},\dots,M_{iN_i})$ denote the vector of post-treatment mediators in cluster $i$.

For notational convenience, define the mediator exposure $S_{ij}=\bigl(M_{ij},G(\underline M_{i(-j)})\bigr)$,
where $M_{ij}$ is the unit's own mediator, $\underline M_{i(-j)}$ is the vector of mediators for the unit's peers, and $G$ is a pre-specified, possibly vector-valued peer-summary map. In the PPACT application below,
\[
\begin{aligned}
G(\underline M_{i(-j)})
&=\left(\overline M_{i,-j},D_{i,-j}\right),\qquad
\overline M_{i,-j}=\frac{1}{N_i-1}\sum_{j'\ne j}M_{ij'},\\
D_{i,-j}
&=\left\{\frac{1}{N_i-2}\sum_{j'\ne j}(M_{ij'}-\overline M_{i,-j})^2\right\}^{1/2}.
\end{aligned}
\]
Thus the exposure records both the average pain impact among a patient's peers and its heterogeneity. The sample SD is well defined because every analyzed PPACT cluster contains at least three patients.
Let $\nu$ be a $\sigma$-finite dominating measure on the support $\calS$ of $S_{ij}$ such that, for each $a\in\{0,1\}$ and $x\in\calX$, the conditional law of $S_{ij}$ given $(A_i=a,X_{ij}=x)$ admits a Radon--Nikod\'ym derivative $f_a(s\mid x)=\frac{dF_{S\mid A=a,X=x}}{d\nu}(s),  s\in\calS$.
This notation unifies the discrete and continuous cases: (i) if $S$ is discrete, then $\nu$ is counting measure and $f_a(s\mid x)=\PP(S=s\mid A=a,X=x)$ and (ii) if $S$ is continuous, then $\nu$ is Lebesgue measure and $f_a(s\mid x)$ is the corresponding conditional density.

For $a\in\{0,1\}$ and $s\in\calS$, let $Y_{ij}(a,s)$ denote the potential outcome for unit $j$ in cluster $i$ under treatment $a$ and mediator exposure $s$.
For $a, a' \in \{0,1\}$, let $S_{ij}(a,a')=\bigl(M_{ij}(a),G(\underline M_{i(-j)}(a'))\bigr)$ denote the cross-world mediator exposure that combines unit $j$'s own potential mediator under treatment $a$ with the peer-summary mediator exposure under treatment $a'$. The natural exposure under uniform treatment $a$ is the diagonal $S_{ij}(a):=S_{ij}(a,a)$, which is the realized exposure when the entire cluster receives treatment $a$; the off-diagonal cases $S_{ij}(1,0)$ and $S_{ij}(0,1)$ are cross-world quantities that figure into the spillover and individual mediation effects defined below. The unit-level observed data are $O_{ij}=(X_{ij},A_i,N_i,S_{ij},Y_{ij})$; equivalently, the cluster record is $O_i=\{V_i,A_i,N_i,(C_{ij},S_{ij},Y_{ij})_{j=1}^{N_i}\}$.
All expectations below are taken with respect to the superpopulation law under a two-stage sampling scheme formalized in the following assumption.
\begin{assumption}[Super-population framework]
\label{asmp:superpopulation}
    (a) The cluster size $N_i$ follows an unknown distribution $ \mathcal{P}^{N}$ over a finite  support on $ \mathbb{N}^{+} $. (b) Conditional on $ N_i $, the joint  distribution $ \mathcal{P}^{V,A,C,S,Y \mid N} $ can be decomposed into $ \mathcal{P}^{V \mid N} \times \mathcal{P}^{C \mid V,N} \times \mathcal{P}^{S \mid C,V,N} \times \mathcal{P}^{Y \mid S,C,V,N} \times \mathcal{P}^{ A }$. Furthermore, positivity holds such that the conditional density $f_{S \mid C, V, A, N} (s \mid c, v, a, n) >0$ for any $\{ s, c, v, a, n \}$ over their valid support.
\end{assumption}
\noindent
Under Assumption \ref{asmp:superpopulation}, a cluster index is drawn uniformly over the $I$ randomized clusters and, given the cluster, a unit index is drawn uniformly over the $N_i$ units in that cluster. Each cluster therefore contributes equally to every induced functional regardless of its size, and we adopt it as the target scale throughout. We suppress $(i,j)$ whenever no confusion can arise.

\subsection{Estimands and assumptions}

Our primary cross-world estimand is the cross-world treatment mean under treatment $1$ with the mediator exposure set to its natural value under treatment $0$, $\theta = \E\{Y(1,S(0,0))\}.$
Let $\delta_a = \E\{Y(a,S(a,a))\}$
denote the arm-specific observed means for $a\in\{0,1\}$. The natural direct and indirect effects on the additive scale are $\mathrm{NDE}=\theta-\delta_0$, $\mathrm{NIE}=\delta_1-\theta,$
so that $\mathrm{TE}=\delta_1-\delta_0=\mathrm{NDE}+\mathrm{NIE}$.
In the PPACT application, $\delta_1$ and $\delta_0$ are the mean 12-month RMDQ scores under the CBT and usual care arms, allowing each arm to generate its own 6-month pain impact profile. The cross-world mean $\theta$ is an \emph{a priori counterfactual} mean. It evaluates the 12-month disability outcome under CBT while holding the patient's own and peer pain impact exposure to the usual care regime. Thus, $\mathrm{NIE}$ represents the part of the contrast of CBT and usual care attributed to shifting the 6-month pain impact exposure from usual care to CBT, whereas $\mathrm{NDE}$ represents the remaining contrast not explained by that measured pain impact shift. Since lower RMDQ scores indicate less disability, negative effects correspond to improvement.

To further disentangle the indirect pathway, we follow \citet{Cheng2026,ohnishi2025} and decompose the natural indirect effect into a \emph{spillover mediation effect} (SME) and an \emph{individual mediation effect} (IME). Define the cross-world mean $\omega = \E\{Y(1,S(1,0))\}$.
The spillover and individual mediation effects on the additive scale are $\mathrm{SME}=\delta_1-\omega = \E\{Y(1,S(1,1))\}-\E\{Y(1,S(1,0))\}$ and $\mathrm{IME}=\omega-\theta = \E\{Y(1,S(1,0))\}-\E\{Y(1,S(0,0))\}$,
so that $\mathrm{NIE}=\mathrm{SME}+\mathrm{IME}$. The IME captures the indirect effect operating through unit $j$'s own mediator $M_{ij}$ while holding the peer mediators fixed at its control value $\underline M_{i(-j)}(0)$, and the SME captures the indirect effect operating through the peer mediators $\underline M_{i(-j)}$ while holding unit $j$'s own mediator fixed at its treated value $M_{ij}(1)$. Identification of $\omega$ requires more than the assumptions used for $\theta$, because $S(1,0)$ mixes potential mediators from two distinct counterfactual treatment regimes within the same cluster.
In PPACT, the IME is the pathway through a patient's own 6-month PEGS response, and the SME is the pathway through the peer PEGS distribution---its mean and variability---in the same PCP cluster. The spillover term should therefore be read as a peer-environment contrast, not necessarily as evidence of direct interaction among all patients in a cluster.

We invoke the following identification assumptions. Consistency (Assumption~\ref{ass:bridge_consistency}) and treatment randomization (Assumption~\ref{ass:treatment_randomization}) are standard in CRTs and are directly supported by the design. Positivity (Assumption~\ref{ass:bridge_positivity}) is a substantive requirement on the observed data law that is not automatically delivered by cluster-level randomization, since randomization does not control the mediator-exposure mechanism $f_a(s\mid x)$; we therefore state it as a separate condition that must be checked in the relevant empirical support.

\begin{assumption}[Consistency]
\label{ass:bridge_consistency}
For all $a\in\{0,1\}$ and $s\in\calS$, $A=a \implies S(a,a)=S$, $A=a,\ S=s \implies Y(a,s)=Y$.
\end{assumption}

\begin{assumption}[Positivity]
\label{ass:bridge_positivity}
For the trial-level treatment assignment, $0<\pi:=\PP(A=1)<1$. For each $a\in\{0,1\}$ and each $(x,s)$ in the support of $(X,S)$ relevant to the causal estimand, $f_a(s\mid x) > 0$.
\end{assumption}

\begin{assumption}[Treatment randomization]
\label{ass:treatment_randomization}
For all $a, a', a''\in\{0,1\}$ and $s\in\calS$, \\
$\bigl\{Y(a,s),\,S(a',a''),\,X\bigr\}\indep A$.
\end{assumption}

Consistency relates the observed mediator exposure and outcome to their potential-outcome counterparts at the realized treatment: first, $S_{ij}=S_{ij}(A_i,A_i)$, so the observed mediator exposure equals the natural (diagonal) potential mediator exposure at the assigned arm; and second, $Y_{ij}=Y_{ij}\bigl(A_i,\,S_{ij}\bigr)$, so the observed outcome equals the potential outcome at the assigned arm and the realized mediator exposure. In PPACT, this means that the recorded CBT or usual care assignment, together with the observed 6-month own and peer PEGS exposure, defines the relevant intervention-mediator history for the 12-month RMDQ outcome.
Positivity requires sufficient overlap in the own and peer pain impact exposure profiles across treatment arms after conditioning on baseline covariates; otherwise, the mediation contrasts would rely on extrapolating to PEGS profiles not represented in the relevant arm. Assumption~\ref{ass:treatment_randomization} holds for PPACT because provider clusters were randomized before patient recruitment. Standard mediation analyses further rely on \emph{sequential ignorability}, which assumes that, conditional on the observed confounders and treatment, the mediator is independent of the potential outcomes. However, because this assumption imposes a cross-world structure and the mediator is not itself randomized, it is not directly empirically verifiable. The sensitivity analysis developed below therefore focuses entirely on violations of sequential ignorability.

\begin{remark}[Two layers of unidentifiability for SME and IME]
\label{rem:two_layers}
The bridge score sensitivity analysis in Section~\ref{sec:bridge_score_sensitivity_crt} varies $\theta$, $\mathrm{NIE}$, and $\mathrm{NDE}$ away from the sequential ignorability anchor through the aggregated bridge score correction $\bar{\Delta}$. The finer decomposition $\mathrm{NIE}=\mathrm{SME}+\mathrm{IME}$ requires identification of the additional cross-world mean $\omega=\E\{Y(1,S(1,0))\}$, which involves a within-cluster joint of $(M_{ij}(1),\underline M_{i(-j)}(0))$ that is never observed and is therefore not identified from the data. This is a separate source of unidentifiability: the bridge score layer addresses unmeasured mediator-outcome confounding within an arm, whereas $\omega$ additionally requires modeling cross-world dependence between potential mediators across arms. Section~\ref{sec:cross_world_omega} introduces a Gaussian copula partial identification of $\omega$ that treats this cross-world dependence as a sensitivity parameter, paralleling the copula modeling of the post-treatment confounder in \citet{ohnishi2026identification}.
\end{remark}

\section{Sensitivity analyses for cross-world assumptions in CRTs}
\label{sec:sensitivity_crt}

This section collects the two sensitivity analyses used in the PPACT application. The first addresses violations of mediator sequential ignorability for the cross-world treatment mean $\theta=\E\{Y(1,S(0,0))\}$. This layer is the CRT adaptation of the bridge score sensitivity analysis of \citet{ohnishi2026sensitivityanalysiscausalmediation}. The individual-level mediator $M$ in that paper is replaced by the CRT mediator exposure $S_{ij}=(M_{ij},G(\underline M_{i(-j)}))$. The second addresses the between-individual cross-world mediator dependence required for the mixed exposure $S_{ij}(1,0)=(M_{ij}(1),G\{\underline M_{i(-j)}(0)\})$ and hence for the spillover/individual split of the natural indirect effect. These two layers are distinct, since the bridge score layer corrects residual mediator-outcome confounding within an arm, whereas the copula layer specifies how potential mediators under different treatment arms are coupled across individuals in the same cluster.

\subsection{Bridge score sensitivity analysis for mediator ignorability}
\label{sec:bridge_score_sensitivity_crt}

Let $f_a(s\mid x)$ denote the treatment-specific conditional law of the CRT mediator exposure $S$ given covariates $X=x$ in arm $a$. Following \citet{ohnishi2026sensitivityanalysiscausalmediation}, define the covariate-indexed bridge score by
\begin{equation}
B_x(s)=\bigl(B_{0x}(s),B_{1x}(s)\bigr)=\bigl(f_0(s\mid x),f_1(s\mid x)\bigr).
\label{def:bridge_score}
\end{equation}
The bridge score retains the two mediator-exposure laws needed for the cross-world mean $\theta$. The balancing and projection properties of this score, as well as the sharp additive sensitivity bound used below, are proved in \citet{ohnishi2026sensitivityanalysiscausalmediation}; here we use their direct CRT specialization with $s=(m,g_\mu,g_\sigma)$, where $(g_\mu,g_\sigma)=G(\underline m_{-j})$ are the peer mediator mean and SD.
For $a\in\{0,1\}$, define the bridge score sensitivity function
\begin{equation*}
\Delta_a(s,b)
=
\E\{Y(1,s)\mid A=a,S=s,B_X(s)=b\}
-
\E\{Y(1,s)\mid A=a,B_X(s)=b\}.
\end{equation*}
It measures the residual mediator-outcome confounding induced by conditioning on the observed mediator exposure $S=s$ within arm $a$, after the mediator-exposure laws have been compressed into $B_X(s)$. Exact sequential ignorability corresponds to $\Delta_0\equiv\Delta_1\equiv0$. Let $\mu_1(s,b)=\E(Y\mid A=1,S=s,B_X(s)=b)$ be the treated-arm observed outcome regression. Under the consistency, positivity, and treatment-randomization assumptions in Section~\ref{sec:bridge_score_sa}, \citet{ohnishi2026sensitivityanalysiscausalmediation} showed that the bridge score representation gives
\begin{align}
\theta
=
\E\!\left[
\int_{\calS}
\Bigl\{
\mu_1\bigl(s,B_X(s)\bigr)-\Delta_1\bigl(s,B_X(s)\bigr)+\Delta_0\bigl(s,B_X(s)\bigr)
\Bigr\}
f_0(s\mid X)\,\nu(ds)
\right].
\label{eq:bridge_theta_identified}
\end{align}

The sharp additive variance bound of \citet{ohnishi2026sensitivityanalysiscausalmediation} is indexed by two interpretable latent-confounding quantities. In a bridge score stratum, let $U$ denote a latent mediator-outcome confounder and let $\psi_a(u;s,b)=\E\{Y(1,s)\mid A=a,B_X(s)=b,U=u\}$. Define the local latent mean-variance component
\begin{equation}
v_a(s,b)
=
\Var\{\psi_a(U;s,b)\mid A=a,B_X(s)=b\}
\label{eq:v_local}
\end{equation}
and the chi-square residual-selection component
\begin{equation}
\chi_a(s,b)
=
\E\!\left[
\left\{
\frac{
dF_{U\mid A=a,S=s,B_X(s)=b}
}{
dF_{U\mid A=a,B_X(s)=b}
}(U)-1
\right\}^2
\mid A=a,B_X(s)=b
\right].
\label{eq:chi_local}
\end{equation}
Then
\begin{equation}
\bigl|\Delta_a(s,b)\bigr|
\le
\Xi_{a,\chi}(s,b)
:=
\sqrt{v_a(s,b)\chi_a(s,b)}.
\label{eq:sharp_additive_bound}
\end{equation}
For calibration it is often easier to place residual selection on a likelihood-ratio scale. Let
\begin{equation}
\gamma_a(s,b)
=
\operatorname*{sup}_{u}
\frac{
dF_{U\mid A=a,S=s,B_X(s)=b}
}{
dF_{U\mid A=a,B_X(s)=b}
}(u)
\ge 1
\label{eq:gamma_local}
\end{equation}
so that $\chi_a(s,b)\le \gamma_a(s,b)-1$. Here $v_a$ measures outcome relevance of the latent confounder through variation in the $U$-specific treated potential-outcome mean, while $\gamma_a$ measures residual selection on the latent confounder induced by conditioning on $S=s$ after conditioning on the bridge score. The likelihood-ratio version of the variance envelope is therefore
\begin{equation}
\bigl|\Delta_a(s,b)\bigr|
\le
\Xi_{a,\gamma}(s,b)
:=
\sqrt{v_a(s,b)\{\gamma_a(s,b)-1\}}.
\label{eq:sharp_additive_lr_bound}
\end{equation}

For the PPACT application, we use the residual-budget calibration from \citet{ohnishi2026sensitivityanalysiscausalmediation}. The goal is not to estimate the unidentified sensitivity functions, but to put their aggregate contribution on a transparent applied scale. Define the treated-arm residual outcome scale
\begin{equation}
\sigma_{\mathrm{res}}^2(s,b)
=
\Var\{Y\mid A=1,S=s,B_X(s)=b\}.
\label{eq:sigma_eta_residual}
\end{equation}
The auxiliary residual model conditions on the mediator exposure and the bridge features, so $\sigma_{\mathrm{res}}(s,b)$ represents the remaining local treated-arm outcome variation available to latent mediator-outcome confounding. This calibration is tied to observed data: at every posterior draw we estimate the treated-arm mean and residual-variance regressions from the observed PPACT outcomes and evaluate the fitted residual scale at the same counterfactual mediator draws used for g-computation. We then impose the residual variance-budget restriction
\begin{equation}
v_a(s,b)\le k_a\sigma_{\mathrm{res}}^2(s,b),
\qquad a\in\{0,1\},
\label{eq:eta_calibration}
\end{equation}
where $k_a\in[0,1]$ controls the share of this local residual variance allowed for latent conditional mean heterogeneity. We also impose the residual-selection restriction $\gamma_a(s,b)\le g_a$ for $a\in\{0,1\}$,
where $g_a\ge1$ is varied over a substantive grid. We examine both an envelope grid and a tipping-point analysis that moves the prior mode across the admissible corrected-bias interval. Combining these restrictions with \eqref{eq:sharp_additive_lr_bound} yields the scalar per-arm envelope
\begin{equation*}
\Xi_{a,\mathrm{RB}}(s,b;k_a,g_a)
=
\widehat\sigma_{\mathrm{res}}(s,b)\sqrt{k_a(g_a-1)},
\qquad a\in\{0,1\}.
\end{equation*}
Let $[L_Y,U_Y]$ denote the known outcome support; for the scaled RMDQ outcome, $[L_Y,U_Y]=[0,1]$. When this support information is used, the residual-budget envelope is intersected pointwise with the support-implied range cap $(U_Y-L_Y)(g_a-1)/g_a$.

Finally, \eqref{eq:bridge_theta_identified} is written in terms of the sequential ignorability anchor and aggregated sensitivity functionals:
\begin{equation}
\theta_{\mathrm{SI}}
:=
\E\!\left[
\int_{\calS}
\mu_1\bigl(s,B_X(s)\bigr) f_0(s\mid X)\,\nu(ds)
\right], \quad
\bar{\Delta}_a
:=
\E\!\left[
\int_{\calS}
\Delta_a\bigl(s,B_X(s)\bigr)f_0(s\mid X)\,\nu(ds)
\right]
\label{eq:theta_SI_def}
\end{equation}
Thus the sensitivity analysis for the natural indirect/direct decomposition can be summarized through the scalar corrected-bias term
\begin{equation}
\theta=\theta_{\mathrm{SI}}+\bar{\Delta}_0-\bar{\Delta}_1
=\theta_{\mathrm{SI}}+\bar{\Delta},
\label{eq:theta_scalar}
\end{equation}
where $\bar{\Delta}:=\bar{\Delta}_0-\bar{\Delta}_1$.
This scalar reduction is useful for the applied analysis. The observed data model learns $\theta_{\mathrm{SI}}$, $\delta_0$, and $\delta_1$, while the sensitivity module specifies how much posterior uncertainty and directional bias should be assigned to $\bar{\Delta}$.
The corrected-bias term in \eqref{eq:theta_scalar} is then bounded by
$|\bar{\Delta}| \le \bar{\Xi}_0+\bar{\Xi}_1 =: \bar{\Xi}_\Delta$, where $\bar{\Xi}_a$ is obtained by averaging the pointwise envelope $\Xi_{a,\mathrm{RB}}\{s,B_X(s);k_a,g_a\}$ over the same $f_0$ integration measure that defines $\bar{\Delta}_a$.

\paragraph{Triangular prior on $\bar{\Delta}$.}
\label{sec:triangular_prior_delta}
For each fixed sensitivity specification $(k_a,g_a)$, we encode directionality with a triangular prior supported on the calibrated interval $[-\bar{\Xi}_\Delta,\bar{\Xi}_\Delta]$. Let
\[
Z=\frac{\bar{\Delta}+\bar{\Xi}_\Delta}{2\bar{\Xi}_\Delta}\in[0,1],
\qquad
Z\sim\mathrm{Triangular}(0,1,q),
\qquad q\in[0,1],
\]
where $q$ is the normalized location of the mode. Equivalently,
$\bar{\Delta}=-\bar{\Xi}_\Delta+2\bar{\Xi}_\Delta Z$ has mode
$\bar{\Xi}_\Delta(2q-1)$, with
\begin{equation}
\E(\bar{\Delta}\mid\bar{\Xi}_\Delta,q)
=\frac{\bar{\Xi}_\Delta(2q-1)}{3},
\qquad
\mathrm{Var}(\bar{\Delta}\mid\bar{\Xi}_\Delta,q)
=\frac{2\bar{\Xi}_\Delta^2(1+q^2-q)}{9}.
\label{eq:delta_mean_var}
\end{equation}
The single parameter $q$ is directly interpretable: $q=1/2$ is the symmetric neutral default, $q<1/2$ favors negative corrected bias, and $q>1/2$ favors positive corrected bias. The boundary choices $q=0$ and $q=1$ give decreasing and increasing right-triangular defaults, respectively. Because positive $\bar{\Delta}$ increases $\theta$, it decreases $\mathrm{NIE}=\delta_1-\theta$ and increases $\mathrm{NDE}=\theta-\delta_0$; negative $\bar{\Delta}$ does the reverse. Table~\ref{tab:prior_interpretation} summarizes this mapping.

\begin{table}[t]
\centering
\caption{Interpretation of triangular-prior modes for the corrected-bias term
$\bar{\Delta}=\bar{\Delta}_0-\bar{\Delta}_1$ entering the mediation decomposition in \eqref{eq:theta_scalar}.}
\label{tab:prior_interpretation}
\begin{adjustbox}{width=\textwidth}
\begin{tabular}{p{0.22\textwidth} p{0.22\textwidth} p{0.56\textwidth}}
\hline
Mode specification & Prior direction on $\bar{\Delta}$ & Interpretation and implication \\
\hline
$q=1/2$
&
$\bar{\Delta}$ is centered at zero
&
Neutral triangular default. The prior is symmetric and does not favor positive or negative aggregated residual mediator-outcome confounding. \\

$q>1/2$
&
Prior mass favors $\bar{\Delta}>0$
&
Positive corrected bias is favored. Moving $q$ toward one shifts mass upward, tending to make $\mathrm{NIE}$ more negative and $\mathrm{NDE}$ less negative. \\

$q<1/2$
&
Prior mass favors $\bar{\Delta}<0$
&
Negative corrected bias is favored. Moving $q$ toward zero shifts mass downward, tending to make $\mathrm{NIE}$ less negative and $\mathrm{NDE}$ more negative. \\
\hline
\end{tabular}
\end{adjustbox}
\end{table}

\subsection{Sensitivity analysis for between-individual mediator dependence}
\label{sec:cross_world_omega}

The natural indirect effect can be decomposed as $\mathrm{NIE}=\mathrm{SME}+\mathrm{IME}$ through the mixed exposure mean $\omega=\E\{Y(1,S(1,0))\}$.
Unlike $\theta$, this quantity requires the joint law of an individual's treated mediator and the peers' control mediators. The observed data identify the arm-specific mediator marginals, but they do not identify the cross-world coupling of $(M_{ij}(1),\underline M_{i(-j)}(0))$. A common point identification assumption is between-individual cross-world mediator independence, $M_{ij}(1)\indep\underline M_{i(-j)}(0)\mid X_i,N_i$ \citep{ohnishi2025,Cheng2026}. We instead treat this dependence as a sensitivity component.

For this layer, we use the standard mediator ignorability anchor for the outcome regression, adapted to clustered mediator exposures. Recall that $X_i=(V_i,\underline C_i)$ collects the cluster-level covariates and all individual-level baseline covariates in cluster $i$.

\begin{assumption}[Sequential ignorability of the mediator]
\label{ass:si_mediators}
For all $a,a'\in\{0,1\}$, all $i,j$, and all $s$ in the relevant support,
\[
Y_{ij}(a,s)\indep \underline M_i(a')\mid A_i,X_i,N_i.
\]
\end{assumption}

Assumption~\ref{ass:si_mediators} is used here only for the mixed exposure outcome regression in $\omega$. The preceding bridge score layer shows how the natural indirect/direct decomposition can be varied away from this anchor for $\theta$; the copula layer below asks how much additional sensitivity arises from the cross-world mediator coupling itself.

\subsubsection{Gaussian copula model for cross-world dependence}
\label{sec:cross_world_copula}

The mixed exposure mean $\omega$ depends on the cross-world pair
$\{M_{ij}(1),G(\underline M_{i(-j)}(0))\}$. Write $G=(G_\mu,G_\sigma)$ for the peer mean and peer SD. Thus the Gaussian construction is confined to the reduced own-mediator/peer-summary representation, while its low-dimensional same-world correlation block is learned from the observed control arm. Let
\begin{align*}
H_{1j}(m\mid X_i,N_i)
&=\PP\{M_{ij}(1)\le m\mid X_i,N_i\},\\
Q_{0r,j}(g\mid X_i,N_i)
&=\PP[G_r\{\underline M_{i(-j)}(0)\}\le g\mid X_i,N_i],\qquad r\in\{\mu,\sigma\},
\end{align*}
denote identified marginal CDFs, learned from the CA-EDP posterior predictive law. Define the latent probability-integral-transform scores
\[
\begin{aligned}
Z_{1j}&=\Phi_1^{-1}\{H_{1j}(M_{ij}(1))\},\\
Z_{rj}&=\Phi_1^{-1}\bigl[Q_{0r,j}\{G_r(\underline M_{i(-j)}(0))\}\bigr],
\qquad r\in\{\mu,\sigma\}.
\end{aligned}
\]
For each $(i,j)$, we use the Gaussian copula
\begin{equation}
\label{eq:gaussian_copula_M}
\begin{pmatrix}Z_{1j}\\Z_{\mu j}\\Z_{\sigma j}\end{pmatrix}
\biggm|X_i,N_i
\sim
\mathrm{N}_3\!\left\{\mathbf 0,
\boldsymbol\Omega_{ij}=
\begin{pmatrix}
1&\rho_{1\mu}&\rho_{1\sigma}\\
\rho_{1\mu}&1&\kappa_{\mu\sigma}\\
\rho_{1\sigma}&\kappa_{\mu\sigma}&1
\end{pmatrix}\right\}.
\end{equation}
The same-world peer-summary correlation matrix
$R_G=\corr\{(Z_{\mu j},Z_{\sigma j})^\top\mid X_i,N_i\}$, with off-diagonal entry $\kappa_{\mu\sigma}$, is learned from the observed control arm. The unidentified vector
$\boldsymbol\rho_{1G}=(\rho_{1\mu},\rho_{1\sigma})^\top$ directly links the individual's treated mediator to the control peer mean and SD. These are exactly the three latent variables required for $\omega$. The conditional law needed for $\omega$ is
\begin{equation}
\begin{pmatrix}Z_{\mu j}\\Z_{\sigma j}\end{pmatrix}
\biggm| Z_{1j}=z,X_i,N_i
\sim \mathrm{N}_2\{\boldsymbol\rho_{1G}z,
R_G-\boldsymbol\rho_{1G}\boldsymbol\rho_{1G}^\top\}.
\label{eq:summary_copula_conditional}
\end{equation}
Given positive-definite $R_G$, the trivariate correlation matrix is positive definite if and only if
\begin{equation}
\boldsymbol\rho_{1G}^\top R_G^{-1}\boldsymbol\rho_{1G}<1.
\label{eq:summary_copula_pd}
\end{equation}
See Supplementary Material Section \ref{sec:sensitivity_analysis_details} for the derivation.
Thus marginal inference for $\omega$ is driven entirely by the two directly interpretable own-to-summary correlations, with a bivariate prior restricted to the ellipse in \eqref{eq:summary_copula_pd}.

\subsubsection{Partial identification of \texorpdfstring{$\omega$}{omega}}
\label{sec:omega_partial_id}

\begin{assumption}[Gaussian copula cross-world dependence]
\label{ass:gaussian_copula}
For each $(i,j)$, the joint distribution of
$\{M_{ij}(1),G_\mu(\underline M_{i(-j)}(0)),G_\sigma(\underline M_{i(-j)}(0))\}$ admits the Gaussian copula representation \eqref{eq:gaussian_copula_M}, with marginal CDFs and $R_G$ identified from the observed data and cross-world dependence indexed by $\boldsymbol\rho_{1G}$.
\end{assumption}

Under Assumption~\ref{ass:gaussian_copula}, $\omega$ is identified up to the value of $\boldsymbol\rho_{1G}$.

\begin{theorem}[Partial identification of $\omega$ under Gaussian copula cross-world dependence]
\label{thm:omega_partial_id}
Suppose Assumptions~\ref{ass:bridge_consistency}--\ref{ass:treatment_randomization}, \ref{ass:si_mediators}, and \ref{ass:gaussian_copula} hold. Then, for any specified $\boldsymbol\rho_{1G}$ satisfying \eqref{eq:summary_copula_pd},
\begin{equation}
\label{eq:omega_partial_id}
\begin{split}
\omega
\;=\;
\E\Biggl[
\frac{1}{N_i}\sum_{j=1}^{N_i}
&\int\!\!\int
\E\bigl\{Y_{ij}\bigm|A_i=1,M_{ij}=m_j,G(\underline M_{i(-j)})=\boldsymbol g,X_i,N_i\bigr\}\\
&\times
f_{M_{ij}\mid A_i=1,X_i,N_i}(m_j)\,
 q_{G_{ij}(0)\mid M_{ij}(1)=m_j,X_i,N_i;\,\boldsymbol\rho_{1G}}(\boldsymbol g)
\,d\boldsymbol g\,dm_j
\Biggr],
\end{split}
\end{equation}
where $\boldsymbol g=(g_\mu,g_\sigma)$ and $q_{G_{ij}(0)\mid M_{ij}(1),X_i,N_i;\boldsymbol\rho_{1G}}$ is determined by \eqref{eq:summary_copula_conditional}, the identified marginal CDFs, and their observed same-world copula. Consequently, given $\boldsymbol\rho_{1G}$, the spillover and individual mediation effects $\mathrm{SME}=\delta_1-\omega$ and $\mathrm{IME}=\omega-\theta$ are identified.
\end{theorem}

The formula separates the observed data components from the sensitivity component. The treated outcome regression, marginal CDFs, and $R_G$ are learned from the observed data. The operative vector $\boldsymbol\rho_{1G}$ enters only through the conditional control peer-summary law given the individual's treated mediator. Setting $\boldsymbol\rho_{1G}=\mathbf0$ recovers the Gaussian-copula analogue of between-individual cross-world independence. In the empirical analysis, we vary each component over a signed grid while fixing the other at zero and use a symmetric bivariate prior restricted to the positive-definite ellipse.

\section{Empirical analysis of the PPACT pragmatic clinical trial}
\label{sec:empirical_analysis}
We applied the two-layer sensitivity analysis to the Pain Program for Active Coping and Training (PPACT) CRT introduced in Section \ref{sec:motivating_example}, which randomized 106 primary PCP clusters across three Kaiser Permanente regions to the PPACT CBT intervention or usual care. The unit of analysis was the patient. Let $A_i$ denote cluster-level randomized assignment for PCP cluster $i$. We defined the patient-level mediator $M_{ij}$ as the PEGS pain impact score at the 6-month assessment (a continuous measure on the 0--10 scale) for patient $j$ in cluster $i$ and used the three-component exposure $S_{ij}=(M_{ij},\overline{M}_{i,-j},D_{i,-j})$, containing own PEGS and the leave-one-out peer mean and sample SD. The outcome $Y_{ij}$ was the RMDQ score at the 12-month assessment (a continuous pain-related-disability measure on the 0--1 scale). The added SD distinguishes clusters in which peers have the same average pain impact but substantially different degrees of heterogeneity.

\subsection{Model specifications}
For the observed data model (Step 1 of Supplementary Material Section \ref{sec:posterior_inference}), we fit the nested common-atoms enriched Dirichlet process (CA-EDP) of \citet{ohnishi2026identification}. This choice responds to two features of PPACT. First, CA-EDP jointly models the baseline covariate distribution, so ignorable missing baseline covariates can be updated inside posterior inference rather than fixed by preprocessing. Second, its enriched layer separates flexible covariate adjustment from clustering of the response law, reducing the tendency of high-dimensional covariates to generate spurious outcome components.

With $R_{ij}=(M_{ij},Y_{ij})$, a truncated representation of the cluster-conditional joint law used for PPACT is
\begin{equation}
\label{eq:CAEDPmodel}
\begin{split}
&p\{N_i,V_i,(C_{ij},R_{ij})_{j=1}^{N_i}\}\\
&\quad=\sum_{k=1}^{K}\pi_k p(N_i,V_i\mid\lambda_k)
\prod_{j=1}^{N_i}
\left[
\sum_{l=1}^{L}w^{\theta}_{kl}p(R_{ij}\mid A_i,N_i,C_{ij},\theta_l)
\left\{\sum_{q=1}^{Q}w^{\phi}_{klq}p(C_{ij}\mid\phi_q)\right\}
\right].
\end{split}
\end{equation}
Equation~\eqref{eq:CAEDPmodel} displays the two covariate levels explicitly: the outer kernel models $(N_i,V_i)$, whereas the enriched covariate atoms model the individual-level vector $C_{ij}$. The response kernel conditions directly on $(A_i,N_i,C_{ij})$ and depends on $V_i$ through the cluster component $k$ and its response weights $w^\theta_{kl}$. Consequently, the induced unit-level response and mediator laws are indexed by the combined vector $X_{ij}=(V_i,C_{ij})$, consistently with the setup in Section~\ref{sec:bridge_score_sa}.

The response atoms $\theta_l$ and individual-level covariate atoms $\phi_q$ are common across cluster components $k$; only their weights vary. This common-atoms construction permits partial sharing across provider clusters without the degeneracy induced by a nested process that equates entire distributions after a single shared atom. The enriched weights $w^{\phi}_{klq}$ allow units with similar response behavior to retain a flexible distribution for $C_{ij}$. We use Gaussian regression atoms for both continuous responses and multivariate Gaussian atoms for the standardized individual-level baseline covariates. Missing coordinates of $C_{ij}$ are sampled conditionally on their observed coordinates, current covariate atom, and mediator regression contribution at each MCMC iteration.

The mediator regression includes an intercept, $A_i$, $N_i$, and the components of $C_{ij}$: baseline age, sex, number of chronic comorbid conditions, anxiety or depression, number of chronic pain conditions, baseline morphine-equivalent opioid dose, baseline benzodiazepine dispensing, baseline RMDQ, and baseline PEGS. At each MCMC iteration, the current CA-EDP mediator law yields $f_0(S_{ij}\mid X_{ij})=f_0(S_{ij}\mid V_i,C_{ij})$ and $f_1(S_{ij}\mid X_{ij})=f_1(S_{ij}\mid V_i,C_{ij})$. The own-mediator density is evaluated as the CA-EDP mixture density. The peer mean is approximated by a Normal distribution using posterior mixture moments, and the positive peer SD is approximated on the log scale using its first two moment-matched parameters. The outcome regression includes $A_i$, $M_{ij}$, $\overline M_{i,-j}$, $D_{i,-j}$, the two log-bridge features, their product, and interactions between $A_i$ and each exposure component. The bridge features are reconstructed at every Gibbs iteration, so they remain posterior functionals of the mediator law rather than fixed preprocessing variables.

We used empirical Bayes Gaussian prior centers from the corresponding least-squares regressions. To avoid making those data-centered priors overly informative, their coefficient standard deviations were the least-squares standard errors inflated by $\sqrt{n/5}=\sqrt{850/5}\approx13.04$ (equivalently, their covariance matrices were multiplied by $n/5$). We used Normal--Inverse--Wishart base measures for the $C_{ij}$ atoms and unit-shape Gamma priors for the three concentration parameters. The revised preprocessing retained 22 missing entries among the individual-level baseline covariates $C_{ij}$, all of which were updated within the sampler. The truncation levels were $(K,L,Q)=(3,8,3)$, and allocation counts for all three layers were monitored during the retained chain. We discarded 2000 burn-in iterations, ran 4000 additional iterations, and retained every fourth draw, giving 1000 posterior draws. At each retained draw, 10 mediator draws were used for g-computation.

We standardized directly over the covariate distribution modeled by the CA-EDP. At each retained iteration, we generated a posterior-predictive population of 106 clusters. For predictive cluster $i$, we sampled $K_i\sim\operatorname{Categorical}(\boldsymbol\pi)$ and $N_i-3\sim\operatorname{Poisson}(\lambda_{K_i})$, where the shift by three reflects the minimum analyzed cluster size and ensures at least two peers remain to define the leave-one-out peer standard deviation. Then, for each unit, $L_{ij}\sim\operatorname{Categorical}(\boldsymbol w^\theta_{K_i})$, $Q_{ij}\sim\operatorname{Categorical}(\boldsymbol w^\phi_{K_iL_{ij}})$, and $C_{ij}\sim\mathrm N(\boldsymbol\mu_{Q_{ij}},\boldsymbol\Sigma_{Q_{ij}})$. Counterfactual mediator and outcome draws were then averaged with weight $1/106$ per predictive cluster and $1/N_i$ per patient within cluster. Thus the standardization integrates over the fitted joint law of cluster size and baseline covariates, including its posterior uncertainty, while preserving the equal-cluster-weighted target.

This specification yields posterior draws of the observed data law
\[
\Theta_{\mathrm{obs}}^{(t)}
=\Bigl(f_0^{(t)},f_1^{(t)},\mu_1^{(t)},\delta_0^{(t)},\delta_1^{(t)}\Bigr),
\]
where the bridge features entering $\mu_1^{(t)}$ are updated jointly with the mediator law at each iteration.
Within each retained iteration, the residual-budget calibration of Steps~3--4 of Section~\ref{sec:posterior_inference} was fitted to observed treated-arm outcomes and evaluated at the same $S(0,0)$ draws used for g-computation. The posterior mean averaged counterfactual residual scale $\widehat\sigma_{\mathrm{res}}^{(t)}$ was $0.137$, with a 95\% posterior interval of $[0.132,0.142]$. The corresponding global auxiliary residual scale $\widehat\sigma_\eta^{(t)}$ had posterior mean $0.184$, with 95\% interval $[0.181,0.187]$. At the primary calibration $(k_a,g_a)=(0.5,1.5)$, the posterior mean envelope for $\bar\Delta$ was $\bar\Xi_\Delta=0.137$, with 95\% interval $[0.132,0.142]$. These observed-data calibrations supply the scale of the sensitivity analysis; $k_a$, $g_a$, and the triangular mode $q$ remain analyst-chosen.

\subsection{Benchmark results under cross-world assumptions}
We first present benchmark results under sequential ignorability and between-individual cross-world mediator independence.
Under these cross-world assumptions, the posterior mean total effect on the 12-month RMDQ disability scale was $-0.056$ with a 95\% credible interval of $[-0.089,-0.026]$. This equal-cluster-weighted, CA-EDP posterior-predictive standardized estimate is directionally consistent with the crude RMDQ difference of roughly $-0.07$. The posterior mean natural indirect effect was $-0.028$ with a 95\% credible interval of $[-0.061,0.002]$, whereas the posterior mean natural direct effect was $-0.028$ with a 95\% credible interval of $[-0.064,0.008]$. Thus the CA-EDP posterior indicates a negative total effect on disability, while the indirect and direct intervals each narrowly include zero at the no-correction anchor. We treat this decomposition as the starting point for sensitivity analysis rather than as evidence that 6-month pain impact is the sole mechanism of CBT.

At $\boldsymbol\rho_{1G}=\mathbf0$, the posterior mean spillover mediation effect was $-0.004$ with a 95\% credible interval of $[-0.018,0.009]$, while the posterior mean individual mediation effect was $-0.024$ with a 95\% credible interval of $[-0.052,0.002]$; the two pieces add to the natural indirect effect. The intervention lowered 6-month PEGS, PEGS is positively associated with later disability, and the individual pathway carries most of the negative indirect signal, whereas the combined peer-mean/peer-variability pathway remains close to zero.

\begin{figure}[t]
\centering
\includegraphics[width=0.75\textwidth]{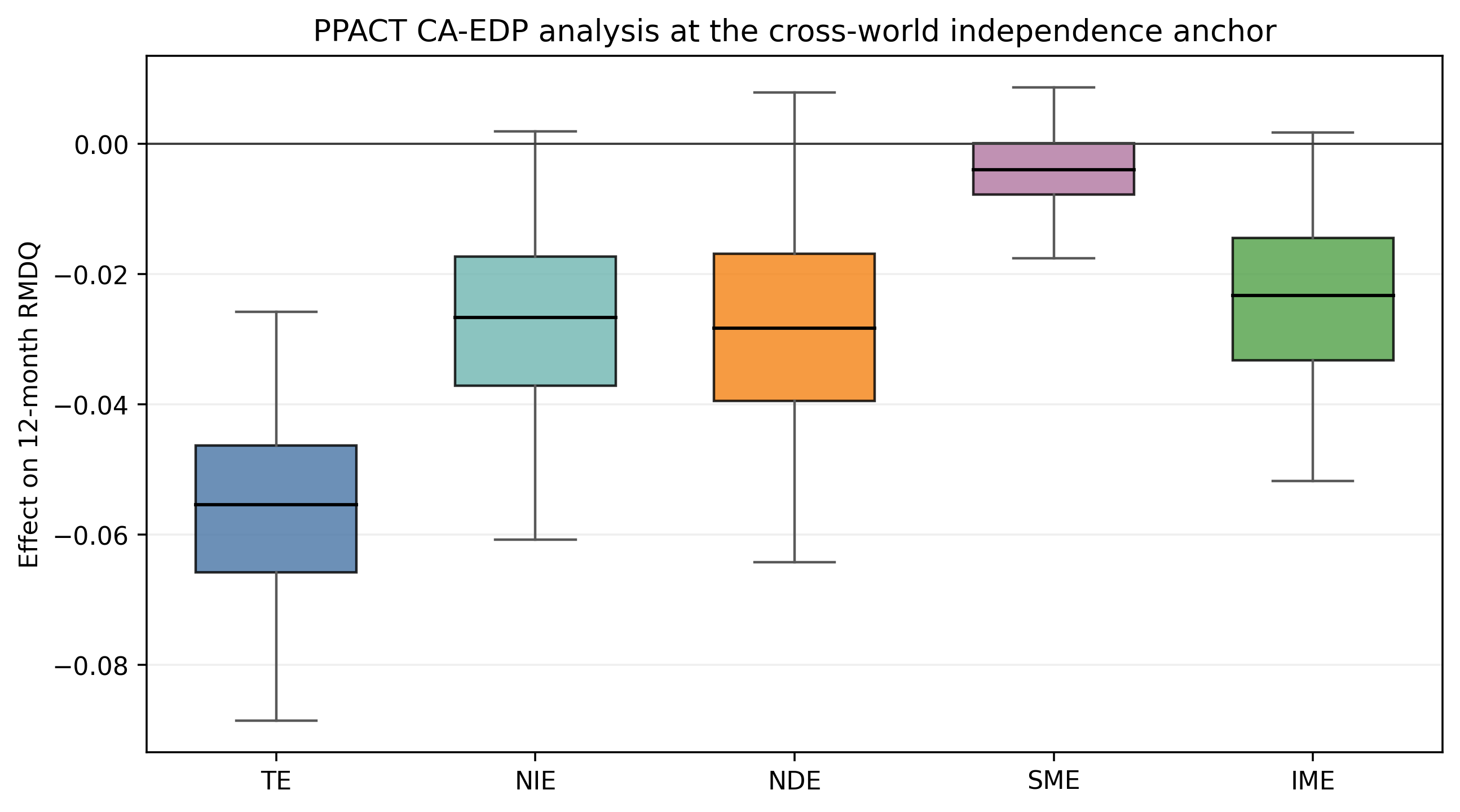}
\caption{CA-EDP posterior distributions for the total, natural indirect, natural direct, spillover mediation, and individual mediation effects under sequential ignorability and $\boldsymbol\rho_{1G}=\mathbf0$. Boxes display interquartile ranges and whiskers display 95\% credible intervals. The total effect is negative; within the indirect pathway, the individual mediation effect carries the negative signal and the spillover effect is centered near zero.}
\label{fig:ppact_default_effects}
\end{figure}

Figure~\ref{fig:ppact_default_effects} summarizes the default posterior with no sensitivity corrections. Only the $95\%$ posterior interval for $\mathrm{TE}$ excludes zero; the intervals for $\mathrm{NIE}$, $\mathrm{NDE}$, $\mathrm{SME}$, and $\mathrm{IME}$ include zero. The posterior probabilities of a negative effect were $1.000$ for $\mathrm{TE}$, $0.958$ for $\mathrm{NIE}$, $0.933$ for $\mathrm{NDE}$, $0.742$ for $\mathrm{SME}$, and $0.965$ for $\mathrm{IME}$. Therefore, the total effect is strongly negative in the posterior, consistent with a CBT-induced reduction in 12-month disability, while neither component of the indirect pathway is separated from zero at the cross-world independence anchor.

\subsection{Sensitivity analysis for sequential ignorability}
\label{sec:empirical_analysis_SI}
We next applied the bridge score sensitivity analysis using the observed-data residual calibration in Section~\ref{sec:bridge_score_sensitivity_crt}. Because $\mathrm{TE}=\delta_1-\delta_0$ depends only on observed means, it is unchanged by this module, so the analysis focuses on $\mathrm{NIE}$ and $\mathrm{NDE}$. We varied two inputs in turn. First, we fixed $(k_a,g_a)=(0.5,1.5)$ and moved the normalized triangular-prior mode $q$ from zero to one, providing a tipping-point analysis over the full range of default triangular directions. Second, we fixed the neutral mode $q=0.5$ and varied $(k_a,g_a)$ over a $3\times3$ grid.

\subsubsection{Triangular-prior mode for \texorpdfstring{$\bar{\Delta}$}{Delta}}
For the first block, $(k_a,g_a)=(0.5,1.5)$ corresponds to half of the bridge-score residual outcome variation budgeted to confounding and a 50\% maximum latent-density tilt induced by $S=s$ after conditioning on $B_X(s)=b$. We varied $q\in\{0,0.25,0.5,0.75,1\}$. The middle value is the symmetric neutral default; moving toward zero favors negative $\bar\Delta$, and moving toward one favors positive $\bar\Delta$. Table~\ref{tab:ppact_sign_probs} reports posterior sign probabilities, and Figure~\ref{fig:ppact_sensitivity} displays the corresponding distributions. The posterior probability $\PP(\mathrm{TE}<0\mid\text{data})$ remains one throughout because the sensitivity correction only reallocates the total effect.

\begin{table}[t]
\centering
\caption{Posterior probability that $\mathrm{NIE}$ and $\mathrm{NDE}$ fall below zero under the sequential-ignorability anchor and five triangular-prior modes $q$ at $(k_a,g_a)=(0.5,1.5)$. Values below zero on the RMDQ disability scale correspond to improvements.}
\label{tab:ppact_sign_probs}
\small
\begin{tabular}{lc rr}
\hline
Prior specification & $q$ & $\PP(\mathrm{NIE}<0\mid \text{data})$ & $\PP(\mathrm{NDE}<0\mid \text{data})$ \\
\hline
Sequential ignorability (no correction) & -- & 95.8\% & 93.3\% \\
Decreasing triangular                 & $0$    & 36.0\% & 84.2\% \\
Mode one-quarter                      & $0.25$ & 50.0\% & 81.7\% \\
Symmetric triangular default          & $0.50$ & 65.7\% & 67.9\% \\
Mode three-quarters                   & $0.75$ & 78.5\% & 48.4\% \\
Increasing triangular                 & $1$    & 81.7\% & 37.7\% \\
\hline
\end{tabular}
\end{table}

\begin{figure}[t]
\centering
\includegraphics[width=\textwidth]{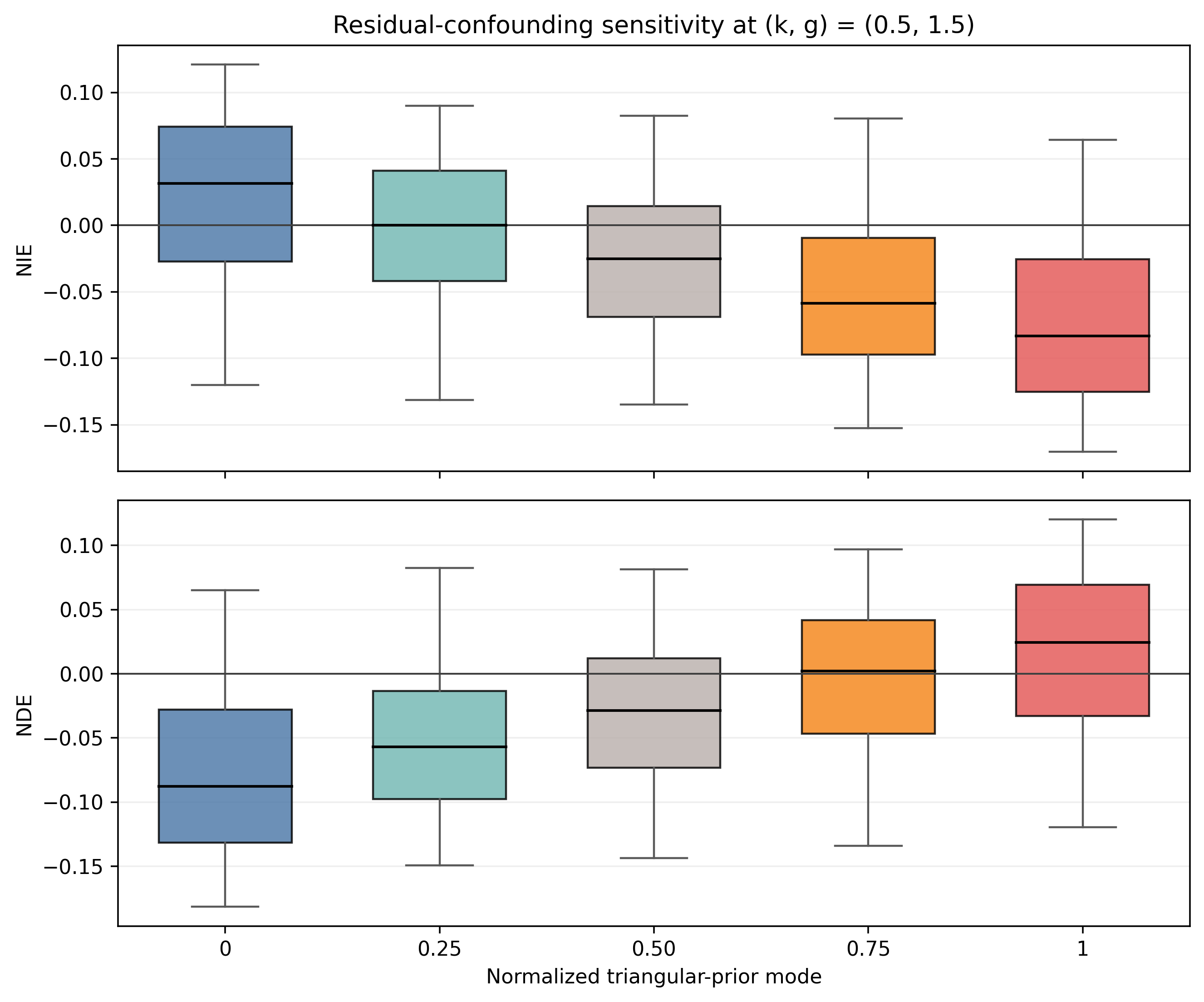}
\caption{Tipping-point sensitivity analysis over the normalized triangular-prior mode $q$ at fixed $(k_a,g_a)=(0.5,1.5)$. The top panel shows $\mathrm{NIE}$ and the bottom panel shows $\mathrm{NDE}$; whiskers display 95\% credible intervals.}
\label{fig:ppact_sensitivity}
\end{figure}

At the neutral mode $q=0.5$, the posterior means remain close to their anchors, but the 95\% intervals widen to $[-0.135,0.082]$ for $\mathrm{NIE}$ and $[-0.144,0.081]$ for $\mathrm{NDE}$; their negative sign probabilities fall to $0.657$ and $0.679$. Directional modes reveal tipping behavior. At $q=0$, the posterior mean $\mathrm{NIE}$ becomes positive ($0.020$) and $\PP(\mathrm{NIE}<0)=0.360$, while the mean $\mathrm{NDE}$ becomes more negative ($-0.077$). At $q=1$, the mean $\mathrm{NIE}$ is $-0.072$ and $\PP(\mathrm{NIE}<0)=0.817$, whereas the mean $\mathrm{NDE}$ becomes positive ($0.016$) and its negative sign probability falls to $0.377$. Thus, plausible directional latent confounding can move posterior support between the indirect and direct pathways.

\subsubsection{Envelope width for \texorpdfstring{$\bar{\Delta}$}{Delta}}
For the second block, we held the triangular prior at its neutral mode $q=0.5$ and varied $k_a\in\{0.25,0.5,1\}$ and $g_a\in\{1.25,1.5,2\}$. The corner $(0.25,1.25)$ gives the tightest envelope, while $(1,2)$ gives the most conservative envelope. Table~\ref{tab:ppact_envelope_signs} reports posterior sign probabilities, and Figure~\ref{fig:ppact_envelope} displays the posteriors.

\begin{table}[t]
\centering
\caption{Posterior probabilities that $\mathrm{NIE}$ and $\mathrm{NDE}$ fall below zero under the neutral triangular prior $q=0.5$, across the $(k_a,g_a)$ envelope grid. The sequential-ignorability row reports the uncorrected anchor.}
\label{tab:ppact_envelope_signs}
\small
\begin{tabular}{cc cc}
\hline
$k_a$ & $g_a$ & $\PP(\mathrm{NIE}<0\mid \text{data})$ & $\PP(\mathrm{NDE}<0\mid \text{data})$ \\
\hline
\multicolumn{2}{l}{Sequential ignorability} & 95.8\% & 93.3\% \\
$0.25$ & $1.25$ & 81.7\% & 78.8\% \\
$0.25$ & $1.50$ & 74.0\% & 73.8\% \\
$0.25$ & $2.00$ & 69.0\% & 67.5\% \\
$0.50$ & $1.25$ & 74.9\% & 73.7\% \\
$0.50$ & $1.50$ & 65.7\% & 67.9\% \\
$0.50$ & $2.00$ & 65.1\% & 62.1\% \\
$1.00$ & $1.25$ & 68.6\% & 68.4\% \\
$1.00$ & $1.50$ & 62.6\% & 64.1\% \\
$1.00$ & $2.00$ & 61.3\% & 59.7\% \\
\hline
\end{tabular}
\end{table}

\begin{figure}[t]
\centering
\includegraphics[width=\textwidth]{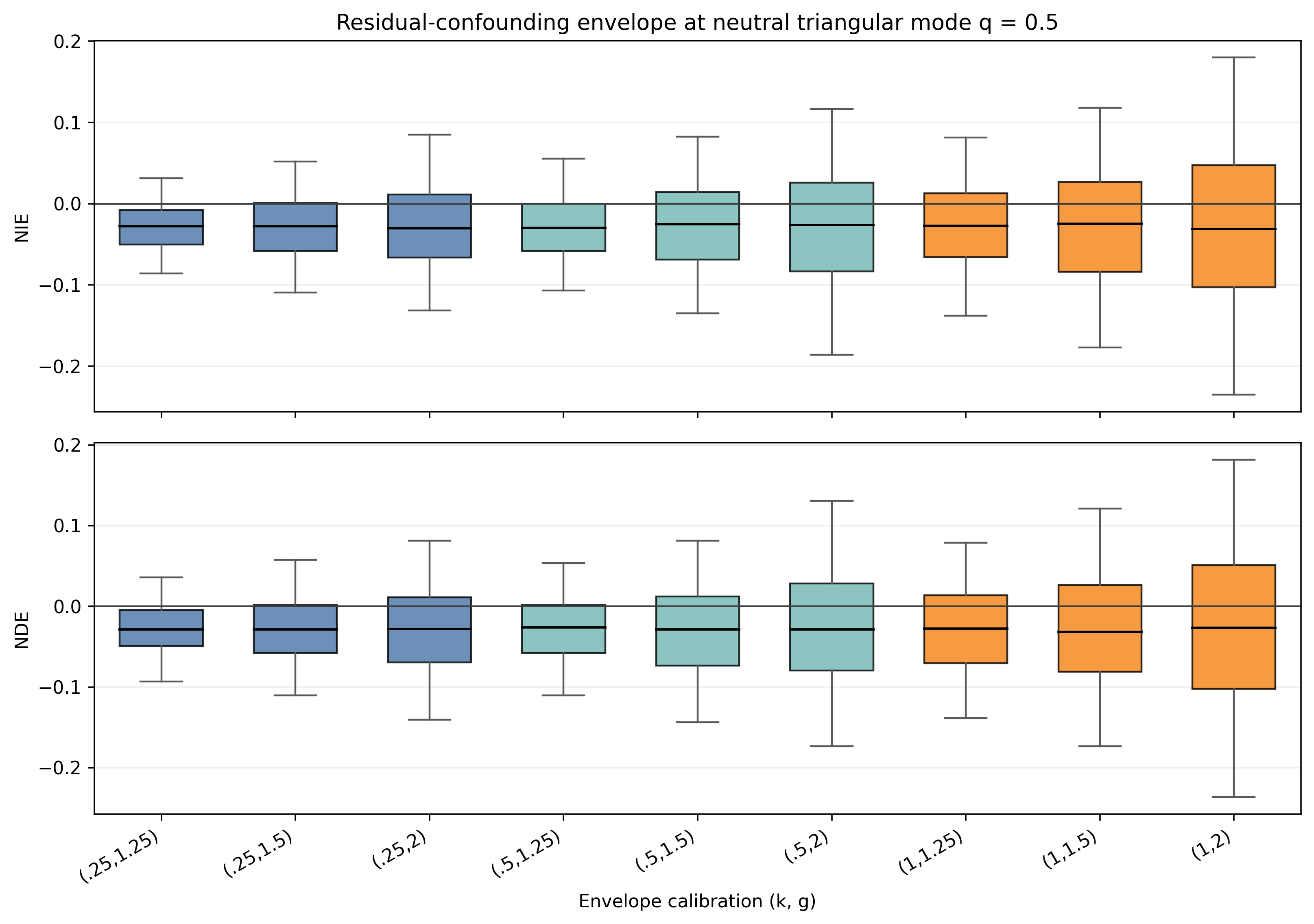}
\caption{Sensitivity analysis over envelope width $(k_a,g_a)$ at the neutral triangular mode $q=0.5$. The top panel shows $\mathrm{NIE}$ and the bottom panel shows $\mathrm{NDE}$; whiskers display 95\% credible intervals.}
\label{fig:ppact_envelope}
\end{figure}

Under the neutral prior, support for negative $\mathrm{NIE}$ and $\mathrm{NDE}$ generally declines as the envelope widens. At $(0.25,1.25)$, the sign probabilities are $0.817$ and $0.788$; at $(1,2)$ they fall to $0.613$ and $0.597$. Corrected posterior medians remain negative throughout the grid, but every interval includes zero. Envelope widening alone therefore does not reverse the median decomposition, yet even the tight calibration leaves appreciably more uncertainty than the no-correction anchor.

Taken together, these results indicate that the uncorrected finding of a small negative mediated effect through 6-month pain impact and a moderate negative residual direct effect is fragile to plausible directional violations of sequential ignorability and robust only to mild envelope widening. The data constrain the cumulative bias $\bar{\Delta}$, but they do not by themselves identify a robust signed conclusion about which pathway carries the CBT effect on 12-month disability.

\subsection{Sensitivity analysis for cross-world mediator independence}
\label{sec:empirical_analysis_copula}
We next applied the trivariate summary-level Gaussian copula of Section~\ref{sec:cross_world_omega}, holding $\bar\Delta=0$ so that variation is attributable only to cross-world mediator dependence. The posterior mean (95\% interval) of the control same-world latent correlation between the peer mean and peer SD was $-0.098$ ($[-0.162,-0.013]$). We varied each component of $\boldsymbol\rho_{1G}=(\rho_{1\mu},\rho_{1\sigma})^\top$ over
\[
\{-0.75,-0.50,-0.25,0,0.25,0.50,0.75\}
\]
while setting the other component to zero. We also assigned a uniform prior over the part of $[-0.75,0.75]^2$ satisfying the positive-definite ellipse in \eqref{eq:summary_copula_pd}. The mixed mean $\omega$ uses the same treated own-mediator draw as $\delta_1$ and changes only its peer-summary coordinates, keeping the spillover contrast focused on the joint peer-distribution pathway.
The trial cannot identify the signs of these cross-world correlations, but the application offers a weak substantive guide. Shared provider practice and cluster patient mix could make patients with high treated PEGS tend to belong to clusters whose peers would also have high PEGS under usual care, suggesting $\rho_{1\mu}>0$. There is no similarly compelling direction for $\rho_{1\sigma}$: greater underlying pain burden could either increase peer heterogeneity or shift the whole peer distribution together without increasing its spread. Negative values represent compensatory or regression-to-the-mean relationships across the own and peer counterfactuals. We therefore report both signs for both components, regard positive $\rho_{1\mu}$ as somewhat more clinically plausible, and avoid using that intuition as identification information.

\begin{table}[t]
\centering
\caption{Posterior summaries of spillover and individual mediation effects as dependence with the control peer mean ($\rho_{1\mu}$) or peer SD ($\rho_{1\sigma}$) is varied, with the other component fixed at zero. The final row uses a bivariate uniform-box prior restricted to the positive-definite ellipse. The bridge correction is held at $\bar\Delta=0$.}
\label{tab:ppact_copula}
\small
\begin{adjustbox}{width=\textwidth}
\begin{tabular}{ll rrrr rrrr}
\hline
& & \multicolumn{4}{c}{$\mathrm{SME}$} & \multicolumn{4}{c}{$\mathrm{IME}$} \\
Varied component & value & mean & 2.5\% & 97.5\% & $\PP(<0)$ & mean & 2.5\% & 97.5\% & $\PP(<0)$ \\
\hline
$\rho_{1\mu}$ & $-0.75$ & -0.003 & -0.020 & 0.014 & 0.652 & -0.025 & -0.054 & 0.003 & 0.957 \\
$\rho_{1\mu}$ & $-0.50$ & -0.003 & -0.019 & 0.011 & 0.669 & -0.024 & -0.053 & 0.002 & 0.959 \\
$\rho_{1\mu}$ & $-0.25$ & -0.004 & -0.018 & 0.010 & 0.712 & -0.024 & -0.052 & 0.002 & 0.962 \\
$\rho_{1\mu}$ & $0.25$  & -0.004 & -0.018 & 0.008 & 0.764 & -0.023 & -0.050 & 0.002 & 0.962 \\
$\rho_{1\mu}$ & $0.50$  & -0.005 & -0.020 & 0.009 & 0.763 & -0.023 & -0.050 & 0.003 & 0.957 \\
$\rho_{1\mu}$ & $0.75$  & -0.006 & -0.021 & 0.009 & 0.770 & -0.022 & -0.050 & 0.004 & 0.941 \\
\addlinespace
$\rho_{1\sigma}$ & $-0.75$ & -0.008 & -0.026 & 0.011 & 0.831 & -0.020 & -0.052 & 0.010 & 0.886 \\
$\rho_{1\sigma}$ & $-0.50$ & -0.007 & -0.022 & 0.010 & 0.822 & -0.021 & -0.051 & 0.006 & 0.927 \\
$\rho_{1\sigma}$ & $-0.25$ & -0.005 & -0.019 & 0.008 & 0.799 & -0.022 & -0.052 & 0.004 & 0.945 \\
$\rho_{1\sigma}$ & $0.25$  & -0.003 & -0.017 & 0.011 & 0.655 & -0.025 & -0.053 & 0.001 & 0.968 \\
$\rho_{1\sigma}$ & $0.50$  & -0.001 & -0.018 & 0.014 & 0.552 & -0.027 & -0.055 & -0.001 & 0.977 \\
$\rho_{1\sigma}$ & $0.75$  &  0.000 & -0.020 & 0.018 & 0.470 & -0.028 & -0.058 & 0.001 & 0.971 \\
\addlinespace
Either          & $0$     & -0.004 & -0.018 & 0.009 & 0.742 & -0.024 & -0.052 & 0.002 & 0.965 \\
Bivariate box prior & -- & -0.004 & -0.021 & 0.013 & 0.698 & -0.024 & -0.053 & 0.004 & 0.946 \\
\hline
\end{tabular}
\end{adjustbox}
\end{table}

\begin{figure}[t]
\centering
\includegraphics[width=\textwidth]{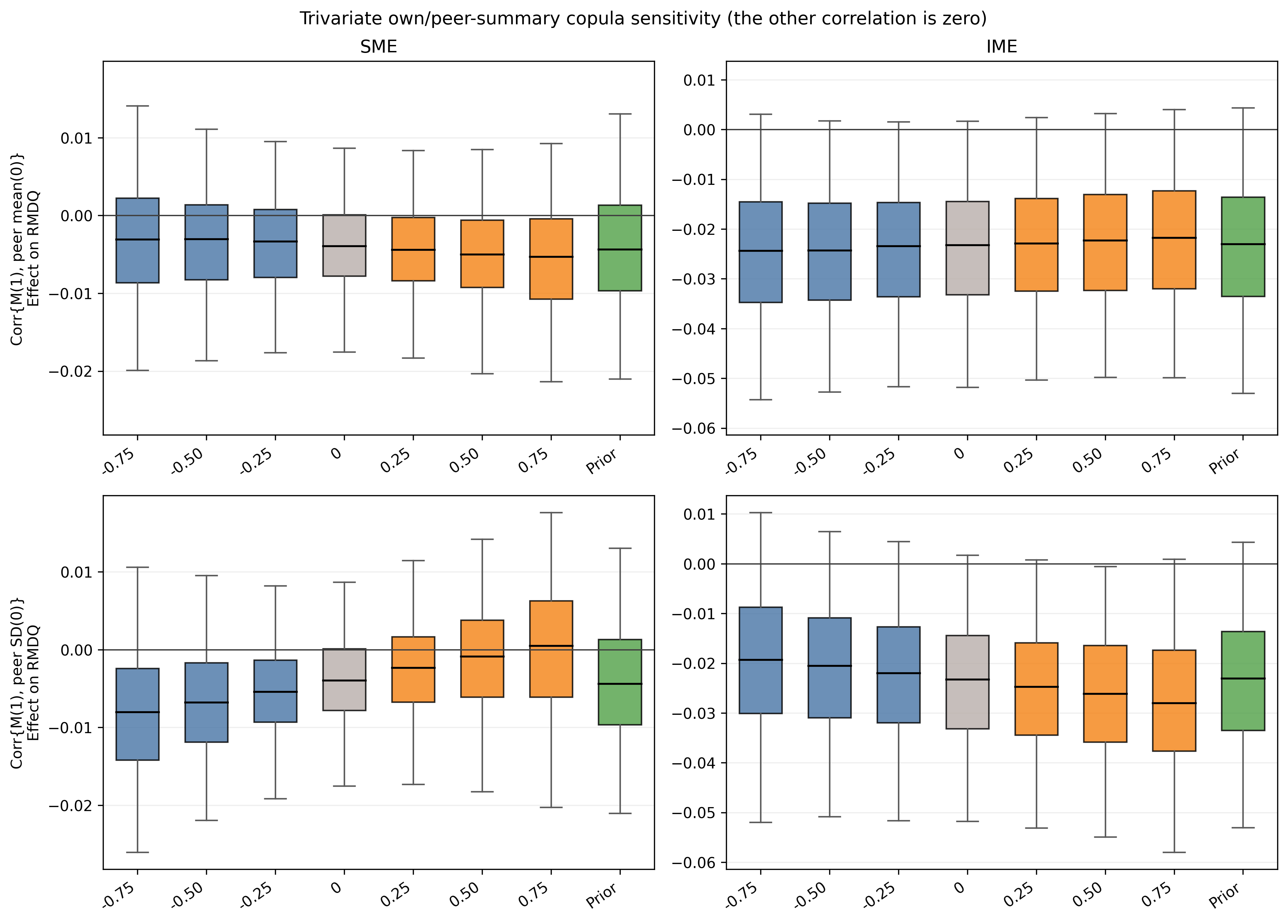}
\caption{Trivariate cross-world summary-copula sensitivity analysis. Rows vary dependence of the treated own mediator with the control peer mean ($\rho_{1\mu}$, top) and peer SD ($\rho_{1\sigma}$, bottom), holding the other component at zero. Columns show $\mathrm{SME}$ and $\mathrm{IME}$; the final point in each panel uses the bivariate box prior for $\boldsymbol\rho_{1G}$ restricted to the positive-definite ellipse. Whiskers display 95\% credible intervals.}
\label{fig:ppact_copula_tipping_point}
\end{figure}

Table~\ref{tab:ppact_copula} and Figure~\ref{fig:ppact_copula_tipping_point} show that peer variability matters enough to inspect but does not alter the substantive split here. Across the peer-mean grid, the mean $\mathrm{SME}$ moves from $-0.003$ at $\rho_{1\mu}=-0.75$ to $-0.006$ at $0.75$; across the peer-SD grid it moves oppositely, from $-0.008$ at $\rho_{1\sigma}=-0.75$ to approximately zero at $0.75$. These small shifts leave every $\mathrm{SME}$ interval covering zero, while the posterior mean $\mathrm{IME}$ remains negative throughout the grid. Under the bivariate prior, the mean $\mathrm{SME}$ is $-0.004$ ($[-0.021,0.013]$) and the mean $\mathrm{IME}$ is $-0.024$ ($[-0.053,0.004]$).

Substantively, the split agrees with the applied PPACT intuition. The intervention reduces 6-month PEGS pain impact on average, higher PEGS is associated with worse 12-month RMDQ disability, and the individual-pathway posterior mean captures most of this negative mediated component. The joint peer distribution contributes little once the individual's own PEGS and baseline covariates are accounted for, and the $\mathrm{SME}$ posterior remains centered near zero even when dependence with peer variability is allowed. We therefore interpret the default split as suggesting that any mediated signal is primarily individual rather than peer-mediated, while retaining the caveat from Section~\ref{sec:empirical_analysis_SI} that the larger natural indirect/direct split remains sensitive to residual mediator-outcome confounding.

\subsection{Joint sensitivity analysis}
\label{sec:empirical_joint}
Finally, we propagated both unidentified components in the same posterior analysis rather than varying them one at a time. We used the neutral triangular prior $q=0.5$ at $(k_a,g_a)=(0.5,1.5)$ and the bivariate uniform-box prior for $\boldsymbol\rho_{1G}$ restricted to its positive-definite ellipse under \eqref{eq:summary_copula_pd}. For each retained CA-EDP draw, these priors were sampled jointly while causal means were standardized over the same posterior-predictive covariate population. The spillover effect remains $\mathrm{SME}=\delta_1-\omega$, while the bridge correction enters $\mathrm{IME}=\omega-(\theta_{\mathrm{SI}}+\bar\Delta)$, preserving $\mathrm{SME}+\mathrm{IME}=\mathrm{NIE}$ draw by draw.

\begin{figure}[t]
\centering
\includegraphics[width=0.80\textwidth]{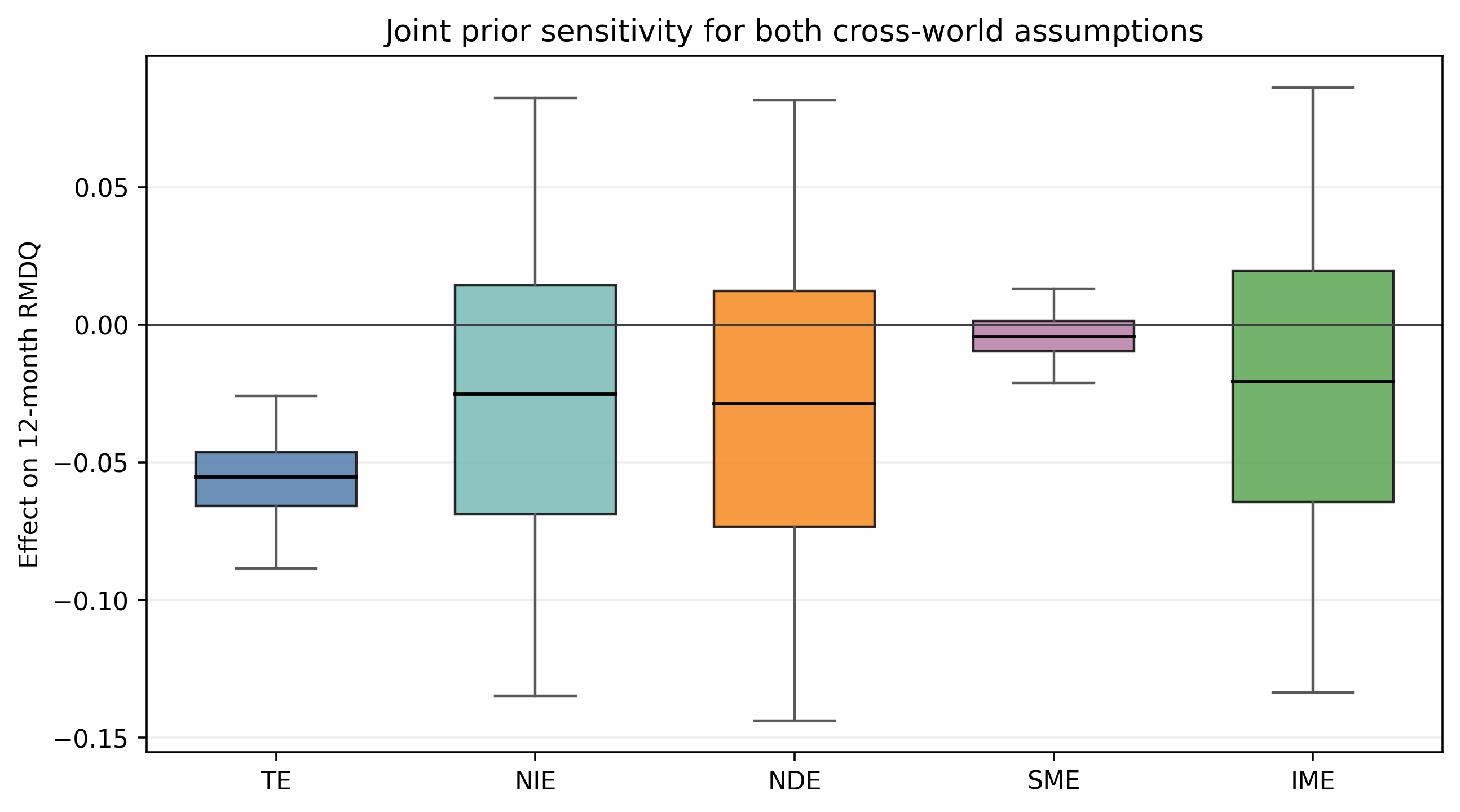}
\caption{Joint prior sensitivity analysis for residual mediator--outcome confounding and cross-world own-to-peer-distribution dependence. The triangular prior uses $q=0.5$ at $(k_a,g_a)=(0.5,1.5)$, and $\boldsymbol\rho_{1G}$ has a uniform prior over the admissible part of $[-0.75,0.75]^2$. Whiskers display 95\% credible intervals.}
\label{fig:ppact_joint_sensitivity}
\end{figure}

Figure~\ref{fig:ppact_joint_sensitivity} shows the resulting posterior. The total effect remains $-0.056$ with 95\% interval $[-0.089,-0.026]$. The joint prior means are $-0.026$ for $\mathrm{NIE}$, $-0.030$ for $\mathrm{NDE}$, $-0.004$ for $\mathrm{SME}$, and $-0.022$ for $\mathrm{IME}$. However, their 95\% intervals are $[-0.135,0.082]$, $[-0.144,0.081]$, $[-0.021,0.013]$, and $[-0.134,0.086]$, respectively. Corresponding negative sign probabilities are $0.657$, $0.679$, $0.698$, and $0.636$. Thus the treatment benefit is stable, but the joint analysis does not support a robust signed conclusion for any individual pathway once both cross-world assumptions are relaxed simultaneously.

\section{Concluding remarks}

We presented sensitivity analyses for causal mediation in the PPACT CRT. The first layer addresses possible violations of sequential ignorability by adapting bridge score sensitivity analysis \citep{ohnishi2026sensitivityanalysiscausalmediation} to $S_{ij}=(M_{ij},\overline M_{i,-j},D_{i,-j})$. Its residual-budget parameters $(k_a,g_a)$ set a data-calibrated envelope, and a single triangular mode $q$ expresses the direction of corrected bias. The second layer addresses cross-world mediator dependence required to split the natural indirect effect. Its trivariate Gaussian copula is limited to the individual's treated mediator and the control peer mean and SD, using the direct own-to-summary vector $\boldsymbol\rho_{1G}$. A CA-EDP observed-data model provides common response atoms, enriched covariate adjustment, posterior covariate updating, and posterior-predictive standardization over the modeled cluster-size and baseline-covariate distribution. Joint draws of $(\bar\Delta,\boldsymbol\rho_{1G})$ yield coherent sensitivity posteriors for both decompositions.

In PPACT, the total effect of CBT on 12-month disability was credibly negative. At the no-correction anchor, the small negative indirect signal was carried by the individual pain-impact pathway rather than peer spillover. Allowing cross-world dependence with peer variability changed the split by an amount comparable to allowing dependence with the peer mean, but neither component moved the spillover interval away from zero. By contrast, the natural indirect/direct decomposition widened sharply under even the neutral triangular residual-confounding prior, and directional modes shifted support between the two pathways. When both assumptions were relaxed jointly, every pathway-specific interval crossed zero even though the total effect remained negative. We therefore interpret the trial as evidence that CBT improved disability and as suggestive of an individual pain-impact mechanism, but not as providing assumption-robust attribution to a specific pathway.


\section*{Funding}
Research in this article was supported by the Patient-Centered Outcomes Research Institute\textsuperscript{\textregistered} (PCORI\textsuperscript{\textregistered} Award ME-2023C1-31350) and National Institutes of Health (R01 HL166324). The statements presented in this article are solely the responsibility of the authors and do not necessarily represent the views of PCORI\textsuperscript{\textregistered}, the National Institutes of Health, or their respective governing bodies.

\begin{singlespace}
\bibliographystyle{plainnat}
\bibliography{literature}

@misc{ohnishi2026sensitivityanalysiscausalmediation,
      title={Sensitivity analysis for causal mediation: bridge score, sharp sensitivity bounds, and calibration}, 
      author={Yuki Ohnishi and Fan Li},
      year={2026},
      eprint={2605.18724},
      archivePrefix={arXiv},
      primaryClass={stat.ME},
      url={https://arxiv.org/abs/2605.18724}, 
}

@article{SjolanderWaernbaum_2024,
url = {https://doi.org/10.1515/jci-2023-0070},
title = {Improved sensitivity bounds for mediation under unmeasured mediator–outcome confounding},
title = {},
author = {Arvid Sjölander and Ingeborg Waernbaum},
pages = {20230070},
volume = {12},
number = {1},
journal = {Journal of Causal Inference},
doi = {doi:10.1515/jci-2023-0070},
year = {2024},
lastchecked = {2026-05-02}
}

@article{le_Cessie_2016,
 ISSN = {10443983, 15315487},
 URL = {https://www.jstor.org/stable/26511892},
 author = {Saskia le Cessie},
 journal = {Epidemiology},
 number = {1},
 pages = {125--132},
 publisher = {Lippincott Williams & Wilkins},
 title = {Bias Formulas for Estimating Direct and Indirect Effects When Unmeasured Confounding Is Present},
 urldate = {2026-05-02},
 volume = {27},
 year = {2016}
}

@article{mccandless2017,
  title={Bayesian sensitivity analysis for unmeasured confounding in causal mediation analysis},
  author={McCandless, Lawrence C and Gustafson, Paul and Levy, Adrian},
  journal={Statistics in Medicine},
  volume={36},
  number={26},
  pages={4169--4183},
  year={2017},
  publisher={Wiley Online Library},
  url={https://onlinelibrary.wiley.com/doi/full/10.1002/sim.70548}
}

@misc{ohnishi2026identification,
      title={Identification and estimation of causal mechanisms in cluster-randomized trials with post-treatment confounding using Bayesian nonparametrics}, 
      author={Yuki Ohnishi and Michael J. Daniels and Lei Yang and Fan Li},
      year={2026},
      eprint={2510.16673},
      archivePrefix={arXiv},
      primaryClass={stat.ME},
      url={https://arxiv.org/abs/2510.16673}, 
}

@article{Cheng2026,
    author = {Cheng, Chao and Li, Fan},
    title = {Semiparametric causal mediation analysis of cluster-randomized trials for indirect and spillover effects},
    journal = {Biometrics},
    volume = {82},
    number = {1},
    pages = {ujag017},
    year = {2026},
    month = {02},
    issn = {0006-341X},
    doi = {10.1093/biomtc/ujag017},
    url = {https://doi.org/10.1093/biomtc/ujag017},
    eprint = {https://academic.oup.com/biometrics/article-pdf/82/1/ujag017/67010844/ujag017.pdf},
}

@article{debar2022ppact,
	author = {DeBar, Lynn and Mayhew, Meghan and Benes, Lindsay and Bonifay, Allison and Deyo, Richard A. and Elder, Charles R. and Keefe, Francis J. and Leo, Michael C. and McMullen, Carmit and Owen-Smith, Ashli and Smith, David H. and Trinacty, Connie M. and Vollmer, William M.},
	title = {A Primary Care-Based Cognitive Behavioral Therapy Intervention for Long-Term Opioid Users With Chronic Pain: A Randomized Pragmatic Trial},
	journal = {Annals of Internal Medicine},
	year = {2022},
	volume = {175},
	number = {1},
	pages = {46--55},
	doi = {10.7326/M21-1436}
}

@article{Ding2016_sharp,
    author = {Ding, Peng and Vanderweele, Tyler J.},
    title = {Sharp sensitivity bounds for mediation under unmeasured mediator-outcome
          confounding},
    journal = {Biometrika},
    volume = {103},
    number = {2},
    pages = {483-490},
    year = {2016},
    month = {04},
    issn = {0006-3444},
    doi = {10.1093/biomet/asw012},
    url = {https://doi.org/10.1093/biomet/asw012},
    eprint = {https://academic.oup.com/biomet/article-pdf/103/2/483/23590467/asw012.pdf},
}

@article{Tchetgen2012,
author = {Eric J. Tchetgen Tchetgen and Ilya Shpitser},
title = {{Semiparametric theory for causal mediation analysis: Efficiency bounds, multiple robustness and sensitivity analysis}},
volume = {40},
journal = {The Annals of Statistics},
number = {3},
publisher = {Institute of Mathematical Statistics},
pages = {1816 -- 1845},
year = {2012},
doi = {10.1214/12-AOS990},
URL = {https://doi.org/10.1214/12-AOS990}
}

@book{Daniels_Linero_Roy_2023
, address={Cambridge}
, title={Bayesian Nonparametrics for Causal Inference and Missing Data}
, doi={10.1201/9780429324222}
, publisher={Chapman and Hall/CRC}
, author={Daniels, Michael J. and Linero, Antonio and Roy, Jason}
, year={2023}}

@article{ohnishi2025,
author = {Yuki Ohnishi and Fan Li},
title = {{A Bayesian nonparametric approach to mediation and spillover effects with multiple mediators in cluster-randomized trials}},
journal = {Journal of the American Statistical Association},
volume = {0},
number = {ja},
pages = {1--20},
year = {2025},
publisher = {ASA Website},
doi = {10.1080/01621459.2025.2544366},
URL = { 
        https://doi.org/10.1080/01621459.2025.2544366
},
eprint = { 
        https://doi.org/10.1080/01621459.2025.2544366
}
}

@article{VanderWeele2014_effect_decomposition,
author = {VanderWeele, Tyler J. and Vansteelandt, Stijn and Robins, James M.},
address = {Philadelphia, PA},
copyright = {Copyright © 2014 Lippincott Williams & Wilkins, Inc},
issn = {1044-3983},
journal = {Epidemiology (Cambridge, Mass.)},
language = {eng},
number = {2},
pages = {300-306},
publisher = {Lippincott Williams & Wilkins, Inc},
title = {Effect Decomposition in the Presence of an Exposure-Induced Mediator-Outcome Confounder},
volume = {25},
year = {2014},
}

@article{Imai2010,
author = {Kosuke Imai and Luke Keele and Teppei Yamamoto},
title = {{Identification, Inference and Sensitivity Analysis for Causal Mediation Effects}},
volume = {25},
journal = {Statistical Science},
number = {1},
publisher = {Institute of Mathematical Statistics},
pages = {51 -- 71},
year = {2010},
doi = {10.1214/10-STS321},
URL = {https://doi.org/10.1214/10-STS321}
}

@article{VanderWeele2013_CRT,
   author = {Tyler J. VanderWeele and Guanglei Hong and Stephanie M. Jones and Joshua L. Brown},
   doi = {10.1080/01621459.2013.779832},
   issn = {01621459},
   issue = {502},
   journal = {Journal of the American Statistical Association},
   pages = {469-482},
   title = {Mediation and spillover effects in group-randomized trials: A case study of the 4Rs educational intervention},
   volume = {108},
   year = {2013},
}

@article{VanderWeele2009,
   author = {Tyler J. VanderWeele},
   doi = {10.1177/0049124110366236},
   issn = {00491241},
   issue = {4},
   journal = {Sociological Methods and Research},
   pages = {515-544},
   title = {Direct and indirect effects for neighborhood-based clustered and longitudinal data},
   volume = {38},
   year = {2009},
}

@article{VanderWeele2010_bias,
  author = {VanderWeele, Tyler J.},
  title = {Bias formulas for sensitivity analysis for direct and indirect effects},
  journal = {Epidemiology},
  volume = {21},
  number = {4},
  pages = {540--551},
  year = {2010}
}

@article{Lindmark2018,
  author = {Lindmark, Anita and de Luna, Xavier and Eriksson, Marie},
  title = {Sensitivity analysis for unobserved confounding of direct and indirect effects using uncertainty intervals},
  journal = {Statistics in Medicine},
  volume = {37},
  number = {10},
  pages = {1744--1762},
  year = {2018}
}
\end{singlespace}

\newpage
\appendix

\section{Proof of Theorem~\ref{thm:omega_partial_id}}
\label{app:omega_proof}
\begin{proof}
The proof proceeds in three steps. Step~1 reduces $\omega$ to a cluster-conditional integral against the cross-world joint of $\{M_{ij}(1),G(\underline M_{i(-j)}(0))\}$. Step~2 identifies the inner conditional expectation as a treated-arm outcome regression. Step~3 uses the Gaussian copula to identify the conditional law of the bivariate control peer summary given $M_{ij}(1)$ up to $\boldsymbol\rho_{1G}$. Throughout, the outer expectation follows the equal-cluster-weighted convention of Section~\ref{sec:bridge_score_sa}.

\medskip
\noindent\textit{Step 1: cluster-conditional decomposition.}
By the law of iterated expectation and the equal-cluster-weighted convention,
\begin{align*}
\omega
=
\E\bigl\{Y(1,S(1,0))\bigr\}
=
\E_{X_i,N_i}\!\left[\frac{1}{N_i}\sum_{j=1}^{N_i}\E\!\left\{Y_{ij}\bigl(1,M_{ij}(1),G(\underline M_{i(-j)}(0))\bigr)\bigm| X_i,N_i\right\}\right].
\end{align*}
By Assumption~\ref{ass:treatment_randomization}, the cross-world potential outcomes and potential mediators are jointly independent of $A_i$ given $(X_i,N_i)$, so
\begin{align*}
&\E\!\left\{Y_{ij}\bigl(1,M_{ij}(1),G(\underline M_{i(-j)}(0))\bigr)\bigm| X_i,N_i\right\}\\
&\qquad =
\E\!\left\{Y_{ij}\bigl(1,M_{ij}(1),G(\underline M_{i(-j)}(0))\bigr)\bigm|A_i=1, X_i,N_i\right\}.
\end{align*}
Write $G_{ij}(0)=G\{\underline M_{i(-j)}(0)\}$. Applying the law of iterated expectation conditional on $(A_i=1,M_{ij}(1)=m_j,G_{ij}(0)=\boldsymbol g)$,
\begin{align}
&\E\!\left\{Y_{ij}\bigl(1,M_{ij}(1),G(\underline M_{i(-j)}(0))\bigr)\bigm|A_i=1, X_i,N_i\right\}\nonumber\\
&\quad=
\int\!\!\int
\E\!\left\{Y_{ij}(1,m_j,\boldsymbol g)\bigm|A_i=1,M_{ij}(1)=m_j,G_{ij}(0)=\boldsymbol g,X_i,N_i\right\}\nonumber\\
&\quad\quad
\times
f_{M_{ij}(1)\mid A_i=1,X_i,N_i}(m_j)\,
q_{G_{ij}(0)\mid A_i=1,M_{ij}(1)=m_j,X_i,N_i}(\boldsymbol g)
\,d\boldsymbol g\,dm_j.
\label{eq:omega_proof_step1}
\end{align}

\medskip
\noindent\textit{Step 2: identification of the inner conditional expectation.}
By Assumption~\ref{ass:si_mediators}, $Y_{ij}(1,m_j,\boldsymbol g)$ is independent of the potential mediator vectors given $(A_i=1,X_i,N_i)$. Therefore, conditioning on $(G_{ij}(0)=\boldsymbol g,M_{ij}(1)=m_j)$ yields the same response surface as conditioning on the corresponding treated-arm exposure:
\begin{align*}
&\E\!\left\{Y_{ij}(1,m_j,\boldsymbol g)\bigm|A_i=1,M_{ij}(1)=m_j,G_{ij}(0)=\boldsymbol g,X_i,N_i\right\}\\
&\quad=
\E\!\left\{Y_{ij}(1,m_j,\boldsymbol g)\bigm|A_i=1,M_{ij}(1)=m_j,G_{ij}(1)=\boldsymbol g,X_i,N_i\right\}.
\end{align*}
By consistency, the latter expectation equals its observed treated-arm counterpart:
\begin{align*}
&\E\!\left\{Y_{ij}(1,m_j,\boldsymbol g)\bigm|A_i=1,M_{ij}(1)=m_j,G_{ij}(1)=\boldsymbol g,X_i,N_i\right\}\\
&\quad=
\E\!\left\{Y_{ij}\bigm|A_i=1,M_{ij}=m_j,G(\underline M_{i(-j)})=\boldsymbol g,X_i,N_i\right\}.
\end{align*}
Consistency and randomization also identify the marginal law of $M_{ij}(1)$ as $f_{M_{ij}\mid A_i=1,X_i,N_i}$.

\medskip
\noindent\textit{Step 3: identification of the cross-world conditional law via the Gaussian copula.}
By Assumption~\ref{ass:treatment_randomization}, the potential mediators are independent of randomized assignment conditional on $(X_i,N_i)$, so
\[
q_{G_{ij}(0)\mid A_i=1,M_{ij}(1)=m_j,X_i,N_i}
=q_{G_{ij}(0)\mid M_{ij}(1)=m_j,X_i,N_i}.
\]
Under Assumption~\ref{ass:gaussian_copula}, define
\[
z_j=\Phi_1^{-1}\{H_{1j}(m_j\mid X_i,N_i)\}.
\]
Equation~\eqref{eq:summary_copula_conditional} gives
\[
\begin{pmatrix}Z_{\mu j}\\Z_{\sigma j}\end{pmatrix}
\biggm|Z_{1j}=z_j,X_i,N_i
\sim
\mathrm N_2\!\left\{\boldsymbol\rho_{1G}z_j,
R_G-\boldsymbol\rho_{1G}\boldsymbol\rho_{1G}^{\top}\right\}.
\]
Pushing its two coordinates through
$Q_{0r,j}^{-1}\{\Phi_1(\cdot)\mid X_i,N_i\}$, $r\in\{\mu,\sigma\}$, determines the joint law
$q_{G_{ij}(0)\mid M_{ij}(1)=m_j,X_i,N_i;\boldsymbol\rho_{1G}}$.
Thus the conditional peer-summary law is identified once $\boldsymbol\rho_{1G}$ is specified.

\medskip
\noindent\textit{Combining the steps.}
Substituting the identification results of Steps~2 and 3 into \eqref{eq:omega_proof_step1} yields, for each fixed $(X_i,N_i)$ and unit $j$,
\begin{align*}
&\E\!\left\{Y_{ij}\bigl(1,M_{ij}(1),G(\underline M_{i(-j)}(0))\bigr)\bigm| X_i,N_i\right\}\\
&\qquad =
\int\!\!\int
\E\bigl\{Y_{ij}\bigm|A_i=1, M_{ij}=m_j,G(\underline M_{i(-j)})=\boldsymbol g,X_i,N_i\bigr\}\\
&\qquad\quad
\times
f_{M_{ij}\mid A_i=1,X_i,N_i}(m_j)\,
q_{G_{ij}(0)\mid M_{ij}(1)=m_j,X_i,N_i;\,\boldsymbol\rho_{1G}}(\boldsymbol g)
\,d\boldsymbol g\,dm_j.
\end{align*}
Averaging over units within cluster and taking the outer expectation over $(X_i,N_i)$ yields the partial identification formula \eqref{eq:omega_partial_id} of Theorem~\ref{thm:omega_partial_id}.

Finally, $\delta_1=\E\{Y(1,S(1,1))\}$ is identified from observed-arm outcomes under Assumption~\ref{ass:treatment_randomization} and Assumption~\ref{ass:bridge_consistency}, and $\theta=\E\{Y(1,S(0,0))\}$ is evaluated through the bridge score representation in \eqref{eq:bridge_theta_identified}. Therefore, given $\boldsymbol\rho_{1G}$ and the specified bridge score sensitivity correction, the spillover and individual mediation effects $\mathrm{SME}=\delta_1-\omega$ and $\mathrm{IME}=\omega-\theta$ are identified as direct consequences of the identification of $\omega$.
\end{proof}

\section{Sensitivity analysis for cross-world dependence}
\label{sec:sensitivity_analysis_details}

\subsection{Derivation of the positive-definiteness condition}

We derive \eqref{eq:summary_copula_pd} using a Schur complement. For compactness, write
$\boldsymbol r=\boldsymbol\rho_{1G}$ and $R=R_G$. Partition the correlation matrix in \eqref{eq:gaussian_copula_M} as
\begin{equation}
\boldsymbol\Omega_{ij}
=
\begin{pmatrix}
1&\boldsymbol r^\top\\
\boldsymbol r&R
\end{pmatrix}.
\label{eq:summary_copula_block_partition}
\end{equation}
The observed peer-summary correlation matrix $R$ is positive definite. Taking its Schur complement in \eqref{eq:summary_copula_block_partition} gives
\begin{equation}
\boldsymbol\Omega_{ij}\succ0
\quad\Longleftrightarrow\quad
1-\boldsymbol r^\top R^{-1}\boldsymbol r>0,
\label{eq:summary_copula_schur}
\end{equation}
which is exactly \eqref{eq:summary_copula_pd}. To display the ellipse explicitly, write
\[
R=
\begin{pmatrix}1&\kappa_{\mu\sigma}\\
\kappa_{\mu\sigma}&1\end{pmatrix},
\qquad
\boldsymbol r=
\begin{pmatrix}\rho_{1\mu}\\\rho_{1\sigma}\end{pmatrix}.
\]
Because $R\succ0$, $|\kappa_{\mu\sigma}|<1$ and
\[
R^{-1}=\frac{1}{1-\kappa_{\mu\sigma}^2}
\begin{pmatrix}1&-\kappa_{\mu\sigma}\\
-\kappa_{\mu\sigma}&1\end{pmatrix}.
\]
Therefore the admissible region can equivalently be written as
\begin{equation}
\frac{\rho_{1\mu}^2-2\kappa_{\mu\sigma}\rho_{1\mu}\rho_{1\sigma}
+\rho_{1\sigma}^2}{1-\kappa_{\mu\sigma}^2}<1.
\label{eq:summary_copula_ellipse_explicit}
\end{equation}
The numerator in \eqref{eq:summary_copula_ellipse_explicit} accounts for the observed association between the peer mean and peer SD, so the two sensitivity correlations cannot generally be chosen independently over the whole square $(-1,1)^2$. Finally, applying the Schur-complement theorem with the upper-left block instead gives
$\boldsymbol\Omega_{ij}\succ0\Longleftrightarrow R-\boldsymbol r\boldsymbol r^\top\succ0$.
This is the conditional covariance in \eqref{eq:summary_copula_conditional}, confirming that the same ellipse is both necessary and sufficient for the conditional simulation used to evaluate $\omega$.

\subsection{Posterior implementation}

The implementation uses the trivariate Gaussian copula construction in \eqref{eq:gaussian_copula_M}, rather than a copula for the potential mediator vectors. For each retained CA-EDP draw, posterior predictive laws under control determine the peer-mean and peer-SD marginal transforms and hence latent scores $(Z_{\mu j},Z_{\sigma j})$. Their equal-cluster-weighted correlation among observed control clusters gives $R_G^{(t)}$; uncertainty in their observed same-world relationship is propagated through the posterior mediator law and posterior updating of incomplete covariates. The peer mean is represented by a Normal approximation using posterior mixture moments. Because the sample SD is positive, its marginal law is represented by a moment-matched lognormal approximation.

For a fixed-grid value or prior draw of $\boldsymbol\rho_{1G}$, the exact bivariate conditional law \eqref{eq:summary_copula_conditional} simulates the control peer mean and SD given the individual's treated mediator. A proposed prior draw is retained only if
$\boldsymbol\rho_{1G}^{\top}R_G^{-1}\boldsymbol\rho_{1G}<1$. The reported one-at-a-time grids vary $\rho_{1\mu}$ or $\rho_{1\sigma}$ over $\{-0.75,-0.50,\ldots,0.75\}$ while setting the other to zero. The joint analysis uses the admissible part of the uniform box $[-0.75,0.75]^2$.

\section{Posterior inference}
\label{sec:posterior_inference}
For completeness, we describe the Bayesian posterior algorithm used for the PPACT sensitivity analysis. The observed data model is learned once from the randomized trial, and the unidentified bridge score correction is supplied by the calibrated sensitivity module described in Section~\ref{sec:bridge_score_sensitivity_crt}.

By \eqref{eq:theta_scalar}, posterior inference for $\theta$ requires posterior draws of the observed data components $f_0$, $f_1$, and $\mu_1$, together with a draw of the corrected-bias term $\bar{\Delta}=\bar{\Delta}_0-\bar{\Delta}_1$. We use a modular construction: the CA-EDP learns the observed data law, while a triangular prior supported on the observed-data-calibrated interval $[-\bar\Xi_\Delta,\bar\Xi_\Delta]$ represents departures from mediator ignorability. A second module varies the unidentified cross-world correlation vector $\boldsymbol\rho_{1G}$ used for the spillover/individual split. Joint sensitivity draws combine both modules within the same retained CA-EDP state and posterior-predictive covariate population.

\subsection{Posterior computation and sampling of the causal estimand}

Let $\Theta_{\mathrm{obs}}$ collect the parameters of the CA-EDP in \eqref{eq:CAEDPmodel}. Its current common response atoms, cluster-size kernel, and enriched covariate weights for $C_{ij}$ induce $f_a^{(t)}(s\mid c)$ and therefore $B_c^{(t)}(s)=\{f_0^{(t)}(s\mid c),f_1^{(t)}(s\mid c)\}$. The bridge features in the outcome design are rebuilt after each mediator-law update.

Write $C_{ij}=(C_{ij,o}^{\top},C_{ij,m}^{\top})^{\top}$ for its observed and missing coordinates. Conditional on covariate atom $q$, standard Gaussian conditioning gives
\[
C_{ij,m}\mid C_{ij,o},q
\sim \mathrm N(\boldsymbol\mu_{m\mid o,q},\boldsymbol\Sigma_{m\mid o,q}).
\]
Within current response atom $l$, write the mediator regression as
$M_{ij}=\alpha_{ij,l}+\boldsymbol\beta_{m,l}^{\top}C_{ij,m}+\varepsilon_{ij,l}$,
where $\alpha_{ij,l}$ contains the intercept, $A_i$, $N_i$, and observed-covariate contribution, and $\varepsilon_{ij,l}\sim\mathrm N(0,\sigma^2_{M,l})$. Combining the covariate atom with this mediator likelihood yields
\begin{align*}
\boldsymbol P_{ij,lq}
&=\boldsymbol\Sigma_{m\mid o,q}^{-1}
+\frac{\boldsymbol\beta_{m,l}\boldsymbol\beta_{m,l}^{\top}}{\sigma^2_{M,l}},\\
C_{ij,m}\mid C_{ij,o},M_{ij},l,q,\cdots
&\sim\mathrm N\!\left[
\boldsymbol P_{ij,lq}^{-1}\left\{
\boldsymbol\Sigma_{m\mid o,q}^{-1}\boldsymbol\mu_{m\mid o,q}
+\frac{\boldsymbol\beta_{m,l}(M_{ij}-\alpha_{ij,l})}{\sigma^2_{M,l}}
\right\},
\boldsymbol P_{ij,lq}^{-1}
\right].
\end{align*}
These draws are made before updating the regression atoms. Thus data augmentation, mediator-law estimation, and outcome regression occur within one posterior chain. Randomization identifies the observed-world means $\delta_a$ without sensitivity parameters.

Because neither sensitivity component is identified from the observed data likelihood, for a fixed triangular specification $(k_a,g_a,q)$ and copula specification $\pi(\boldsymbol\rho_{1G}\mid R_G)$ we use the modular factorization
\[
p(\Theta_{\mathrm{obs}},\bar{\Delta},\boldsymbol\rho_{1G}\mid \mathcal D)
=
p(\Theta_{\mathrm{obs}}\mid \mathcal D)\,
\pi(\bar{\Delta}\mid\Theta_{\mathrm{obs}},k_a,g_a,q)\,
\pi(\boldsymbol\rho_{1G}\mid R_G,\Theta_{\mathrm{obs}}),
\]
Here $\mathcal D$ denotes the observed data. Dependence on $\Theta_{\mathrm{obs}}$ appears only through the observed-data calibration, identified marginals, and $R_G$; the likelihood does not update the analyst-chosen sensitivity directions. The following nested Monte Carlo is a post-processing layer for the retained CA-EDP states: it computes causal estimands and sensitivity draws but does not feed back into, refit, or otherwise change the observed-data posterior. 

\paragraph{Step 1: draw the observed data law.}
For each MCMC iteration $t=1,\dots,T$, draw $\Theta_{\mathrm{obs}}^{(t)} = \Bigl(f_0^{(t)},f_1^{(t)},\mu_1^{(t)},\delta_0^{(t)},\delta_1^{(t)}\Bigr)$
from the CA-EDP posterior. At each retained iteration, draw a predictive standardization population of $I$ clusters from the current covariate law. In the PPACT specialization,
\begin{align*}
\widetilde K_i^{(t)}&\sim\operatorname{Categorical}(\boldsymbol\pi^{(t)}),
&\widetilde N_i^{(t)}-3&\sim\operatorname{Poisson}(\lambda_{\widetilde K_i^{(t)}}^{(t)}),\\
\widetilde L_{ij}^{(t)}&\sim\operatorname{Categorical}
(\boldsymbol w^{\theta,(t)}_{\widetilde K_i^{(t)}}),
&\widetilde Q_{ij}^{(t)}&\sim\operatorname{Categorical}
(\boldsymbol w^{\phi,(t)}_{\widetilde K_i^{(t)},\widetilde L_{ij}^{(t)}}),\\
\widetilde C_{ij}^{(t)}&\sim\mathrm N\!\left(
\boldsymbol\mu_{\widetilde Q_{ij}^{(t)}}^{(t)},
\boldsymbol\Sigma_{\widetilde Q_{ij}^{(t)}}^{(t)}\right).
\end{align*}
This draws from the cluster-size and individual-covariate distribution already present in the observed-data model; no additional Bayesian-bootstrap distribution is introduced.

\paragraph{Step 2: construct the bridge score and evaluate the anchoring estimand.}
For each draw $t$, draw $M$ Monte Carlo (MC) draws of mediators $\widetilde S^{(t,m)}_{ij}$ from $f_0^{(t)}(s\mid \widetilde C_{ij}^{(t)})$ for $m=1,\dots,M$, and compute the bridge score $\widetilde B_{ij}^{(t,m)}=B_{\widetilde C_{ij}^{(t)}}^{(t)}(\widetilde S^{(t,m)}_{ij})$.
Define
\begin{equation}
\theta_{\mathrm{SI}}^{(t,m)} = \frac{1}{I}\sum_{i=1}^I\frac{1}{\widetilde N_i^{(t)}}
\sum_{j=1}^{\widetilde N_i^{(t)}}
\mu_1^{(t)}\bigl(\widetilde S_{ij}^{(t,m)},\widetilde B_{ij}^{(t,m)}\bigr).
\label{eq:theta_SI_draw_new}
\end{equation}
The $M$ MC samples are drawn for g-computation to approximate the integral of $\mu_1$ in \eqref{eq:theta_SI_def}.
The average in \eqref{eq:theta_SI_draw_new} is the equal-cluster-weighted superpopulation expectation fixed in Section~\ref{sec:bridge_score_sa}, evaluated over the posterior-predictive covariate population. The same predictive clusters are used for $\delta_a^{(t)}$, $\omega^{(t)}$, and the counterfactual residual-scale average.

\paragraph{Step 3: calibrate the pointwise residual outcome scale and scalar bounds.}
For each draw $t$, construct the realized log-bridge features for the observed treated units,
\[
\widetilde B_{ij}^{(t)}
=
\bigl\{\log f_0^{(t)}(S_{ij}\mid X_{ij}),\,
\log f_1^{(t)}(S_{ij}\mid X_{ij})\bigr\},
\qquad A_i=1.
\]
Stabilizing through logs avoids numerical underflow in the auxiliary calibration regression. Let $\mathcal I_1=\{i:A_i=1\}$ and use treated-arm equal-cluster weights $w_{ij}=|\mathcal I_1|^{-1}N_i^{-1}$ for $i\in\mathcal I_1$. Fit an auxiliary calibration regression of $Y_{ij}$ on an intercept, the mediator-exposure components $S_{ij}=(M_{ij},\overline M_{i,-j},D_{i,-j})$, and the log-bridge features $\widetilde B_{ij}^{(t)}$ among treated units, and let $\widehat m_{1,SB}^{(t)}\bigl(S_{ij},\widetilde B_{ij}^{(t)}\bigr)$ denote its fitted value. Including $S_{ij}$ in the regression subtracts the within-bridge score-stratum $S$-driven variation in $Y$, leaving residual variation used to calibrate the local latent mean-variance component $v_a(s,b)$. Define
\[
\overline Y_1
=
\sum_{i\in\mathcal I_1}\sum_{j=1}^{N_i}w_{ij}\,Y_{ij},
\qquad
\widehat V_1
=
\sum_{i\in\mathcal I_1}\sum_{j=1}^{N_i}
w_{ij}\bigl(Y_{ij}-\overline Y_1\bigr)^2,
\]
and
\[
\widehat R_{1,SB}^{2,(t)}
=
1-
\frac{
\sum_{i\in\mathcal I_1}\sum_{j=1}^{N_i}
w_{ij}\bigl\{Y_{ij}-\widehat m_{1,SB}^{(t)}\bigl(S_{ij},\widetilde B_{ij}^{(t)}\bigr)\bigr\}^2
}{
\sum_{i\in\mathcal I_1}\sum_{j=1}^{N_i}
w_{ij}\bigl(Y_{ij}-\overline Y_1\bigr)^2
}.
\]
The residual outcome scale is
\begin{equation}
\widehat\sigma_\eta^{(t)}
=
\bigl[\widehat V_1\{1-\widehat R_{1,SB}^{2,(t)}\}\bigr]^{1/2}.
\label{eq:sigma_eta_draw}
\end{equation}
This global scale is useful as a diagnostic. For the residual-budget envelope itself, fit a low-dimensional working model for the local residual variance using the squared residuals
\[
\widehat e_{ij}^{(t)}
=
Y_{ij}-\widehat m_{1,SB}^{(t)}\bigl(S_{ij},\widetilde B_{ij}^{(t)}\bigr),
\qquad A_i=1,
\]
for example a log-linear regression of $\log\{(\widehat e_{ij}^{(t)})^2+c^{(t)}\}$ on $(1,M_{ij},\overline M_{i,-j},\widetilde B_{ij}^{(t)})$ with a small positive stabilizing constant $c^{(t)}$. Let $\widehat\sigma_{\mathrm{res}}^{(t)}(s,b)$ denote the resulting predicted residual standard deviation. For each counterfactual draw $\widetilde S_{ij}^{(t,m)}$ from Step~2, define the pointwise envelope
\begin{equation}
\Xi_a^{(t,m)}(\widetilde S_{ij}^{(t,m)};k_a,g_a)
=
\widehat\sigma_{\mathrm{res}}^{(t)}\!\left(\widetilde S_{ij}^{(t,m)},\widetilde B_{ij}^{(t,m)}\right)
\sqrt{k_a(g_a-1)},
\qquad
a\in\{0,1\}.
\label{eq:xi_point_draw}
\end{equation}
If the outcome support $[L_Y,U_Y]$ is used, replace \eqref{eq:xi_point_draw} by the pointwise minimum of this residual-budget envelope and $(U_Y-L_Y)(g_a-1)/g_a$. The scalar envelope is then the same equal-cluster-weighted Monte Carlo average used for $\theta_{\mathrm{SI}}$:
\begin{equation}
\bar{\Xi}_a^{(t)}
=
\frac{1}{M}\sum_{m=1}^M
\frac{1}{I}\sum_{i=1}^I\frac{1}{\widetilde N_i^{(t)}}
\sum_{j=1}^{\widetilde N_i^{(t)}}
\Xi_a^{(t,m)}(\widetilde S_{ij}^{(t,m)};k_a,g_a).
\label{eq:xi_bar_draw}
\end{equation}
By construction, this is the posterior-draw version of $\bar{\Xi}_a$, so the per-arm calibrated admissible interval for the aggregated sensitivity functional is $[-\bar{\Xi}_a^{(t)},\bar{\Xi}_a^{(t)}]$, and the corrected-bias term has the worst-case envelope
\begin{equation}
\bar{\Xi}_\Delta^{(t)} = \bar{\Xi}_0^{(t)} + \bar{\Xi}_1^{(t)},
\qquad
|\bar{\Delta}^{(t)}|\le \bar{\Xi}_\Delta^{(t)}.
\label{eq:xi_delta_draw}
\end{equation}

\paragraph{Step 4: draw the bias correction.}
Draw $Z^{(t)}\sim\mathrm{Triangular}(0,1,q)$ for fixed $q$ and set
$\bar{\Delta}^{(t)}=-\bar{\Xi}_\Delta^{(t)}+2\bar{\Xi}_\Delta^{(t)}Z^{(t)}$.
The draw is supported on the calibrated interval, has mode
$\bar\Xi_\Delta^{(t)}(2q-1)$, and has the conditional moments in \eqref{eq:delta_mean_var}.

\paragraph{Step 5: compute the cross-world treatment mean.}
By the scalar bridge score representation in \eqref{eq:theta_scalar}, its posterior draw is
\begin{equation}
\theta^{(t)}
=
\theta_{\mathrm{SI}}^{(t)}
+\bar{\Delta}^{(t)},
\label{eq:theta_draw_new}
\end{equation}
where $\theta_{\mathrm{SI}}^{(t)} = M^{-1}\sum_{m=1}^M \theta_{\mathrm{SI}}^{(t,m)}$ is the Monte Carlo sequential ignorability anchor from Step~2.

\paragraph{Step 6: compute causal effects.}
For each $t$, compute
\[
\mathrm{NDE}^{(t)}
=
\theta^{(t)}-\delta_0^{(t)},
\qquad
\mathrm{NIE}^{(t)}
=
\delta_1^{(t)}-\theta^{(t)},
\qquad
\mathrm{TE}^{(t)}
=
\delta_1^{(t)}-\delta_0^{(t)}.
\]
For the copula component, calculate the control-arm same-world score correlation matrix $R_G^{(t)}$ from the marginally transformed observed control peer mean and peer SD using equal-cluster weights. For each fixed or prior specification $r$, draw $\boldsymbol\rho_{1G,r}^{(t)}$ within the positive-definite ellipse. In nested replicate $m$, transform the simulated treated own mediator in the predictive population to
$\widetilde z_{1ij}^{(t,m)}=\Phi_1^{-1}\{H_{1j}^{(t)}(\widetilde M_{ij}^{1,(t,m)})\}$ and draw
\[
\widetilde{\boldsymbol Z}_{Gij}^{0,(t,m)}
\sim
\mathrm N_2\!\left\{
\boldsymbol\rho_{1G,r}^{(t)}\widetilde z_{1ij}^{(t,m)},
R_G^{(t)}-\boldsymbol\rho_{1G,r}^{(t)}
\boldsymbol\rho_{1G,r}^{(t)\top}
\right\}.
\]
The current CA-EDP predictive marginal laws map this draw to the control peer mean and peer SD, $\widetilde{\boldsymbol G}_{ij}^{0,(t,m)}$. Recompute the bridge features at the mixed exposure
$\widetilde S_{ij}^{10,(t,m)}=\{\widetilde M_{ij}^{1,(t,m)},\widetilde{\boldsymbol G}_{ij}^{0,(t,m)}\}$ and set
\begin{equation}
\omega_r^{(t)}
=
\frac{1}{M}\sum_{m=1}^M
\frac{1}{I}\sum_{i=1}^I\frac{1}{\widetilde N_i^{(t)}}
\sum_{j=1}^{\widetilde N_i^{(t)}}
\mu_1^{(t)}\!\left(\widetilde S_{ij}^{10,(t,m)},\widetilde B_{ij}^{10,(t,m)}\right).
\label{eq:omega_draw_new}
\end{equation}
At the sequential-ignorability anchor,
\[
\mathrm{SME}_r^{(t)}=\delta_1^{(t)}-\omega_r^{(t)},
\qquad
\mathrm{IME}_r^{(t)}=\omega_r^{(t)}-\theta_{\mathrm{SI}}^{(t)}.
\]
For joint sensitivity analysis, retain the same $\omega_r^{(t)}$ but replace $\theta_{\mathrm{SI}}^{(t)}$ by $\theta^{(t)}=\theta_{\mathrm{SI}}^{(t)}+\bar\Delta^{(t)}$ in IME. Hence $\mathrm{SME}^{(t)}+\mathrm{IME}^{(t)}=\mathrm{NIE}^{(t)}$ draw by draw under both the anchor and joint analyses. Posterior means, credible intervals, and posterior probabilities are empirical summaries of the retained draws.

\end{document}